\documentclass[12pt]{amsart}

\usepackage[hmargin=2.5cm,vmargin=2.5cm]{geometry}
\usepackage{graphicx}
\usepackage{amssymb}
\usepackage{color}

\newtheorem{theorem}{Theorem}[section]
\newtheorem{lemma}[theorem]{Lemma}
\newtheorem{proposition}[theorem]{Proposition}
\newtheorem{corollary}[theorem]{Corollary}

\theoremstyle{definition}
\newtheorem{remark}[theorem]{Remark}
\newtheorem{definition}[theorem]{Definition}

\newcommand{\todo}[1]{\textbf{\color{red} ToDo: \emph{ #1\/}}}

\renewcommand{\todo}[1]{}

\makeatletter
\@ifundefined{showcaptionsetup}{}{%
 \PassOptionsToPackage{caption=false}{subfig}}
\usepackage{subfig}
\makeatother

\newcommand\E{\mathcal{E}}
\newcommand\V{\mathcal{V}}
\renewcommand\L{\mathcal{L}} 
\newcommand\Cut{\mathcal{C}}
\newcommand\genN{\mathcal{G}_N} 
\newcommand\Cnct{\mathfrak{C}}

\newcommand\opH{\mathcal{H}}

\newcommand\R{\mathbb{R}}
\newcommand\Q{\mathbb{Q}}
\newcommand\T{\mathbb{T}}
\newcommand\C{\mathbb{C}}
\newcommand\Z{\mathbb{Z}}

\newcommand\ba{\boldsymbol{\alpha}}
\newcommand\bx{\boldsymbol{x}}
\newcommand\Lmat{\boldsymbol{L}}
\newcommand\Smat{\boldsymbol{S}}
\newcommand\Id{\boldsymbol{I}}
\newcommand\tS{\widetilde{S}}

\newcommand\cc[1]{\overline{#1}}

\newcommand\tx{\varkappa} 
\newcommand\gradx{\nabla_\varkappa} 

\newcommand\lv{\vec{l}}
\newcommand\xv{\vec{\tx}}
\newcommand\av{\vec{\alpha}}

\renewcommand\S{\Sigma}
\newcommand\Sreg{\Sigma^{reg}}
\newcommand\Sgen{\Sigma^{g}}
\newcommand\bgm{\mu_{\lv}^{g}} 
\newcommand\CR{\mathcal{R}}

\newcommand\siloc[1]{\sigma^{(#1)}}
\newcommand\beloc[1]{\beta^{(#1)}}

\renewcommand\Pr[1]{p_{#1}}
\newcommand\Prc[2]{\mathbb{P}\left(#1\left|#2\right.\right)}

\newcommand\M{\mathcal{M}} 

\newcommand\inv{\mathcal{I}}
\newcommand\Diff[2]{\frac{d#1}{d#2}}

\DeclareMathOperator{\diag}{diag}
\DeclareMathOperator{\normal}{\hat{n}}
\DeclareMathOperator{\Bin}{Bin}

\begin{document}

\title{Nodal statistics on quantum graphs}

\author{Lior Alon}
\address{Department of Mathematics, Technion --- Israel Institute of
  Technology, Haifa, Israel}
\author{Ram Band}
\address{Department of Mathematics, Technion --- Israel Institute of
  Technology, Haifa, Israel}
\author{Gregory Berkolaiko}
\address{Department of Mathematics, Texas A\&M University, College
  Station, TX 77843-3368, USA}

\begin{abstract}
  It has been suggested that the distribution of the suitably
  normalized number of zeros of Laplacian eigenfunctions contains
  information about the geometry of the underlying domain.  We study
  this distribution (more precisely, the distribution of the ``nodal
  surplus'') for Laplacian eigenfunctions of a metric graph.
  The existence of the distribution is established, along with its
  symmetry.  One consequence of the symmetry is that the graph's first
  Betti number can be recovered as twice the average nodal surplus of
  its eigenfunctions.  Furthermore, for graphs with disjoint
  cycles it is proven that the distribution has a universal form ---
  it is binomial over the allowed range of values of the surplus.  To
  prove the latter result, we introduce the notion of a local nodal
  surplus and study its symmetry and dependence properties,
  establishing that the local nodal surpluses of disjoint cycles
  behave like independent Bernoulli variables.
\end{abstract}

\maketitle

\section{Introduction}
\label{sec:introduction}

Studying various properties of the nodal sets of Laplacian
eigenfunctions is a subject with a long history in mathematical
physics.  The number of the zeros or the nodal domains (depending on
the context and the dimension) of the $n$-th eigenfunction is one of
the simplest quantities one can observe experimentally. Yet, analytical
study of this quantity as a function of $n$ is complicated
by its non-locality, which can be appreciated by observing that
different nodal domains of the same eigenfunction can vary wildly in
size and shape.  Classical results in estimating this number include
those of Sturm \cite{Stu_jmpa36}, Courant \cite{Cou_ngwg23} and
Pleijel \cite{Ple_cpam56}, with notable recent contributions by Ghosh,
Reznikov and Sarnak \cite{GhoRezSar_gafa13} and by Jung and Zelditch
\cite{JunZel_ma16,JunZel_jde16}.  In a series of works of Smilansky
and co-authors
\cite{BluGnuSmi_prl02,GnuSmiWeb_wrm04,GnuSmiSon_jpa05,GnuKarSmi_prl06,BanShaSmi_jpa06,KarSmi_jpa08,BanOreSmi_incoll08}
it has also been proposed that studying the \emph{distribution} of the
appropriately rescaled number of nodal domains can reveal much about
the geometry of the underlying system.  This line of thought has lead
to such results as the Bogomolny and Schmit \cite{BogSch_prl02}
prediction for the average number of the nodal domains (by analogy
with a percolation model), a proof by Nazarov and Sodin
\cite{NazSob_ajm09} that the average rescaled number of the nodal domains for
random waves is non-zero, as well as a slew of high-precision
numerical studies \cite{Nas_thesis11,Kon_thesis12,BelKer_jpa13}.

In this paper we will investigate the
distribution of the nodal count of Laplacian eigenfunctions on
\emph{metric graphs}, a class of models which was used to study the
nodal count distributions from the very start
\cite{GnuSmiWeb_wrm04,BanShaSmi_jpa06,BanOreSmi_incoll08}.  We will
show, first of all, that the statistical distribution of the nodal
count is a well defined object.  The nodal count of the $n$-th
eigenfunction, when shifted down by $n-1$, takes values in a bounded
range of integers; we call the shifted count the \emph{nodal surplus}.
For any graph we will show that the limiting frequency of the
appearance of a given surplus in the spectral sequence can be
calculated as an integral of a piecewise constant function over an
analytic subvariety of a torus which is called \emph{secular
  manifold}.

The nodal count distribution is shown to be symmetric around its mean,
which is equal to half the first Betti number $\beta$ of the graph;
conversely, the first Betti number can be recovered from the nodal
statistics.  Furthermore, for a class of graphs whose cycles are
disjoint, we will prove that, despite knowing neither the individual
eigenvalues nor the zero count of individual eigenfunctions, we can
predict the limiting nodal count distribution.  It takes a universal
form --- the binomial distribution over the a priori allowed range of
values, from 0 to $\beta$.

To prove the latter, we introduce a new concept of a \emph{local nodal
  surplus}.  That such a quantity can be defined at all is very
interesting in itself, due to the issue of non-locality mentioned
above.  To explain this concept informally, recall that the global
nodal surplus can be viewed as a deviation of the number of zeros from
the baseline $n-1$, attributable to the presence of cycles in the
graph.  One therefore expects that the extra number of zeros is
``localized'' on the cycles and, if the graph has block structure (can
be disconnected by cutting bridges, for example), one should be able
to talk about the local nodal surplus of an individual block.  This
notion will be rigorously defined in this paper by analytic means.
Its geometric meaning is far from obvious: the global nodal
surplus is the difference between the number of zeros and $n-1$, and
while the local meaning of the number of zeros is obvious enough,
there is no local analogue of the eigenfunction's number $n$.  Our
analytic definition, however, allows us to prove that for a graph with
disjoint cycles, the local surpluses of the cycles behave like
\emph{independent} Bernoulli variables, hence the binomial
distribution of the global nodal surplus.

\section{Definitions and Main Results}
\label{sec:definitions_results}

Let $\Gamma\left(\V,\E,l\right)$ be a finite connected metric graph
with a set of vertices $\V$ and a set of edges $\E$.  The sizes of the
sets $\V$ and $\E$ are denoted $V$ and $E$ correspondingly.  The last
entry of the triple is the length function $l: \E \to \R_+$ which
associates to each edge $e\in\E$ a positive length which we will
denote $l_e$.  We will identify each edge with an interval $[0,l_e]$
of the corresponding length.  In doing so one needs to choose an
orientation for the edge, but this can be done arbitrarily and does
not affect the results in any way.  Note that we allow multiple edges
between the same pair of vertices and also edges with both ends at the
same vertex (\emph{loops}).

A \emph{quantum graph }is a metric graph $\Gamma\left(\V,\E,l\right)$
equipped with a \emph{Schr\"odinger type operator}
acting on the Hilbert space $\bigoplus_{e\in E} L^2([0,l_e])$ with a
suitable domain.  We will not consider potentials, restricting our
attention to Laplace operator
\begin{equation}
  \label{eq:def_laplacian}
  \opH=-\Delta,\qquad
  \opH:\,f \mapsto -\Diff{^{2}f}{x_{e}^{2}}.
\end{equation}
The \emph{magnetic Schr\"odinger operator} will also play an important
role; we will define it in Section~\ref{sec:mag_nodal}.

In this paper we will consider the most common vertex conditions
for which these operators are self-adjoint.  We say that a function
obeys the \emph{Neumann boundary conditions} if at any vertex $v\in\V$
it is continuous and
\begin{equation}
  \label{eq:current_cons}
  \sum_{e\in E_{v}} \frac{d}{dx}f(v) = 0
\end{equation}
where $E_{v}$ is the set of edges incident to the vertex $v$, and by
convention the derivatives are taken into the edge $e$.  At a
vertex of degree one, the above conditions reduce to the standard
Neumann condition $f'(v)=0$.  At such vertices we will also allow
Dirichlet conditions (i.e.\ $f(v)=0$).  A connected quantum graph
different from a circle or a polygon and with the above vertex
conditions will be called a \emph{nontrivial standard graph}.

Further details about theory of quantum graphs can be found in the
books \cite{GnuSmi_ap06,BerKuc_graphs,Mugnolo_book} as well as the
recent elementary notes \cite{Ber_prep16}.

\subsection{The nodal surplus}
\label{sec:nodal_surplus}

Since our quantum graph is compact, the operator $\opH$ has a discrete
spectrum of eigenvalues $\left\{ \lambda_{n}\right\} _{n=1}^{\infty}$
and corresponding eigenfunctions
$\left\{ f_{n}\right\}_{n=1}^{\infty}$.  For the operators presented
here, the spectrum is non-negative, and we will use the notation
$\lambda_{n}=k_{n}^{2}$. From here on we will also refer to
$k_{n}\ge0$ as the eigenvalue of the graph. The eigenfunctions of
\eqref{eq:def_laplacian} can be chosen to be real and, if the
eigenfunction does not vanish on entire edges (which is possible on
graphs due to failure of unique continuation principle), one can count
the number of the zeros of the $n$-th eigenfunction.  This quantity
will be denoted by $\phi_n$ and will be the main object of our study.

From now on we will call $k_n$ and $f_n$ \emph{generic eigenvalue
  and eigenfunction} if the eigenvalue $k_n$ is simple and the
eigenfunction $f_n$ does not vanish on the vertices (and therefore
edges) of the graph.  We will routinely assume that the edge lengths
are independent over the field $\Q$ of rational numbers (or
\emph{rationally independent}).  This will be shown to guarantee that
a majority of the eigenvalues are generic\footnote{unless the
  graph is a circle or a polygon, which we specifically excluded when
  defining a nontrivial standard graph} (see
\cite{Fri_ijm05,BerLiu_jmaa17} and Appendix~\ref{sec:volume_Sgen}).
Furthermore, if the graph has no loops, for a choice of
rationally independent edge lengths \emph{all} eigenvalues are generic, hence the name.

The nodal count of a tree graph
is $\phi_{n}=n-1$ which is a generalization of Sturm's oscillation
theorem that was obtained in \cite{PokPryObe_mz96,Sch_wrcm06}
(interestingly, the converse result has also been established
\cite{Ban_ptrsa14}: if the nodal count is $\phi_n=n-1$ then the graph
is a tree). For graphs which are not trees $n-1$ provides a baseline
from which the actual number of zeros does not stray very far.
Defining the $n$-th \emph{nodal surplus} by $\sigma_n = \phi_n -
(n-1)$ we have the following bounds
\begin{equation}
  \label{eq:surplus_bounds}
  0 \leq \sigma_n \leq \beta,
\end{equation}
where $\beta$ is the number of independent cycles on the graph (i.e.\
the number of generators of the first homology group of the graph ---
the first Betti number), and is equal to
\begin{equation}
  \beta = E - V + 1.
  \label{eq:def_beta}
\end{equation}
We remark that the $1$ above stands for the number of connected components
of the graph.  The upper bound was proven in \cite{GnuSmiWeb_wrm04}
and the lower bound in \cite{Ber_cmp08} (see also
\cite{BanBerSmi_ahp12}).

An interesting insight on the nature of these bounds comes from the
characterization of the nodal surplus in terms of the Morse index of
the eigenvalue as a function of certain parameters
\cite{Ber_apde13,Col_apde13,BerWey_ptrsa14}.  This theorem will play a
central role in our study and we review it in
Section~\ref{sec:mag_nodal}.

\subsection{Main Results}
\label{sec:main_results}

In order to investigate the nodal surplus of a graph we wish to define
a \textit{surplus distribution} that will give the density of a given
value in the nodal surplus sequence.
\begin{theorem}
  \label{thm:existence_of_distrib}
  Let $\Gamma$ be a nontrivial standard graph with rationally
  independent lengths.  Then the \emph{nodal surplus distribution} of
  $\Gamma$ is a well defined probability distribution on
  the set $\{0,1,\ldots,\beta\}$ given by
  \begin{equation}
    \Pr{s} = \lim_{N\to\infty}
    \frac{\left|\left\{n\in\genN|\,\sigma_n=s\right\}\right|}
    {\left|\genN\right|},
    \qquad 0 \leq s \leq \beta,
    \label{eq:surplus_distr_defn}
  \end{equation}
  where $\genN$ is the set of indices $1 \leq n \leq N$ such that
  $k_n$ is generic. Furthermore, the distribution is symmetric, in the
  sense that
  \begin{equation}
    \Pr{s}=\Pr{\beta-s}, \quad 0\leq s \leq \beta,
    \label{eq:symmetric_distr}
  \end{equation}
  and, therefore, the value of $\beta$ can be recovered as twice the
  average nodal surplus
  \begin{equation}
    \label{eq:recover_beta}
    \beta = 2\,\mathbb{E}(\sigma_n)
    := \lim_{N\to\infty} \frac{2}{\left|\genN\right|}
    \sum_{n\in\genN} \sigma_n.
  \end{equation}
\end{theorem}

\begin{remark}
  Equation~\eqref{eq:recover_beta} generalizes the inverse result of
  \cite{Ban_ptrsa14}: $\Pr{0}=1$ implies $\beta=0$.
\end{remark}

\begin{figure}
  \centering
  \includegraphics[scale=1]{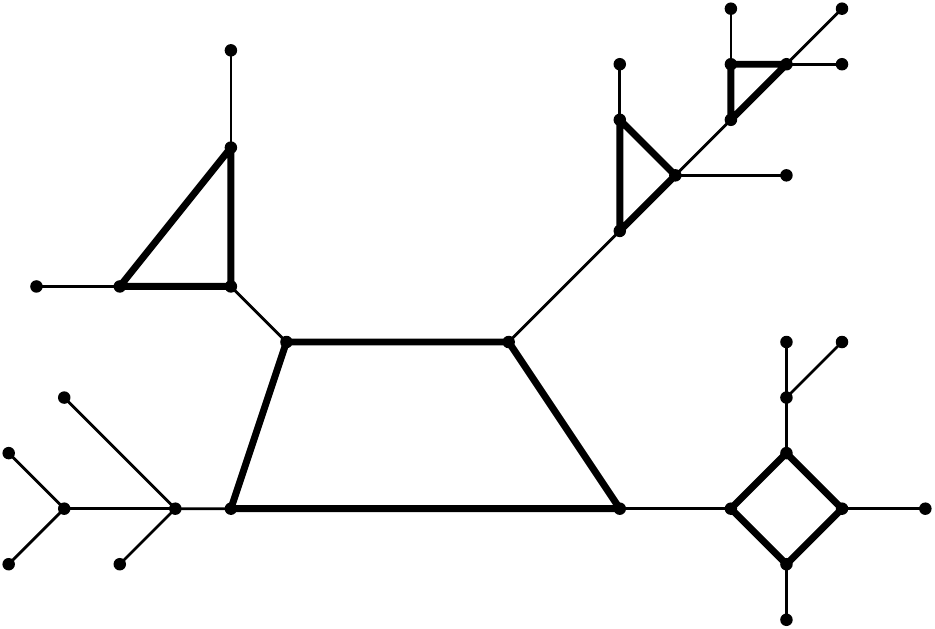}
  \caption{An example of a graph with five disjoint cycles.}
  \label{fig:graph_disjoint_cycles}
\end{figure}

We now define the special structure of the graphs where we can say
more about the form of the distribution $\Pr{s}$.  A \emph{simple
  cycle} in the graph $\Gamma$ is a sequence of vertices
$[v_1, \ldots, v_k]$, such that there is an edge connecting vertex
$v_j$ to $v_{j+1}$ for all $j$ (including $v_k$ to $v_1$) and no
vertex appears more than once.  A graph is said to have \emph{disjoint
  cycles} if there is a basis set of $\beta$ cycles such that each
vertex $v\in \V$ is traversed by at most one cycle, see
Figure~\ref{fig:graph_disjoint_cycles} for an example.

\begin{theorem}
  \label{thm:binomial_distr}
  Let $\Gamma$ be a nontrivial standard graph with rationally
  independent lengths.  If the cycles of $\Gamma$ are disjoint, then
  the nodal surplus distribution of $\Gamma$ is Binomial with parameters
  $\beta$ and $1/2$. That is,
  \begin{equation}
    \Pr{s} =  \frac1{2^{\beta}} \binom{\beta}{s}, \qquad
    0 \leq s \leq \beta.
    \label{eq:binomial_distr}
  \end{equation}
\end{theorem}

Our approach to analyzing the nodal surplus sequence $\{\sigma_n\}$ is
interpreting it as a sample of a certain function defined on a certain
manifold endowed with a probability measure.  The manifold and the
measure go back to an idea of Barra and Gaspard \cite{BarGas_jsp00};
the fact that the nodal surplus can be read off the manifold was shown
by Band \cite{Ban_ptrsa14} who converted the nodal-magnetic connection
of Berkolaiko and Weyand \cite{BerWey_ptrsa14} into a function on the
manifold.  Existence of the limit in (\ref{eq:surplus_distr_defn})
follows from the ergodicity of the sampling process (which goes back
to Weyl \cite{Wey_ma16}). The symmetry (\ref{eq:symmetric_distr}) is a
consequence of a simple symmetry of the nodal surplus function, the
underlying manifold and the measure.

Our second main result requires a much more detailed analysis of the
manifold and the nodal surplus function.  Considering a more general
situation, a graph consisting of disjoint blocks of cycles, we show
that the total nodal surplus is a sum of the ``local surpluses'' of
individual blocks.  These local surpluses also have symmetry similar
to (\ref{eq:symmetric_distr}), but with $\beta$ equal to the number of
cycles in their block.  Moreover this symmetry is independent of the
values taken by other local surpluses.  If each block has just one
cycle, the local surpluses become independent (rather than merely
``independently symmetric''), thus producing the binomial
distribution.

Our proofs require several technical tools and results.  To avoid
tiring the reader we describe these results on the ``as needed''
basis.  Section~\ref{sec:defining_distro} introduces magnetic
Laplacian, the magnetic-nodal connection, secular equation and secular
manifold, and the Barra-Gaspard measure before proceeding to prove
Theorem~\ref{thm:existence_of_distrib}.
Section~\ref{sec:graphs_block_structure} defines the notion of a block
of a graph, introduces scattering from a graph, a factorization of the
secular equation by splitting a graph into two parts, defines the
local nodal surplus and studies its properties, culminating in a proof
of Theorem~\ref{thm:binomial_distr}.  In
Appendices~\ref{sec:volume_Sgen}, \ref{sec:zero_length_edge} and
\ref{sec:proof_splitting} we prove some auxiliary results used in the
paper.  Finally (and perhaps most interestingly for readers looking
for open problems), in Appendix~\ref{sec:examples} we present some
simple examples of nodal surplus distributions, both numerical and
analytical, mostly of the graphs falling outside the assumptions of
Theorem~\ref{thm:binomial_distr}.  This helps us to understand to what
extent the assumptions are optimal.

\section{Defining the distribution}
\label{sec:defining_distro}

To get an analytic handle on the nodal distribution we combine two
techniques of quantum graphs analysis: the magnetic-nodal connection
of Berkolaiko, Colin de~Verdi\`ere and Weyand
\cite{Ber_apde13,Col_apde13,BerWey_ptrsa14} and the secular manifold
of Barra and Gaspard \cite{BarGas_jsp00} further developed in
\cite{BerWin_tams10,BanBer_prl13,CdV_ahp15}.  We lay out the required
foundations in the next subsections.

\subsection{The magnetic Laplacian and the magnetic-nodal connection}
\label{sec:mag_nodal}


The \emph{magnetic Laplacian} is the operator on
$\bigoplus_{e\in E} L^2([0,l_e])$ acting as
\begin{equation}
  \opH_A : f \mapsto \left(i\Diff{}{x}+A(x)\right)^{2}f,
  \label{eq:def_magnetic_laplacian}
\end{equation}
where $A(x)$ is a piecewise continuous 1-form called the \emph{magnetic
potential}. The vertex conditions are modified by substituting
\eqref{eq:current_cons} with
\begin{equation}
  \label{eq:current_cons_mag}
  \sum_{e\in E_v} \left(i\Diff{}{x}+A(x)\right)f(v) = 0.
\end{equation}
Naturally, when $A\equiv0$ we recover the non-magnetic Laplacian of
\eqref{eq:def_laplacian}.

A \emph{magnetic flux} of a magnetic potential $A$
through an oriented cycle $\gamma$ is defined as
\begin{equation}
  \alpha_{\gamma} := \left[\oint_{\gamma}A(x)\right]\mod2\pi.
\end{equation}
The fluxes through independent cycles of the graph $\Gamma$ completely
determine the magnetic Laplacian in the following sense.

\begin{lemma}
  Two magnetic Laplacian operators $\opH_{A},\opH_{A'}$ are unitarily
  equivalent if the magnetic potentials $A$ and $A'$ have the same
  flux through every cycle $\gamma$.
\end{lemma}

The above lemma is well-known with proofs in the quantum graph setting
appearing, for example, in \cite{KosSch_cmp03,BerWey_ptrsa14}.  Due to
additivity of fluxes, it is enough to know them for a fundamental set
of cycles, $\beta$ in number.  Fixing a particular fundamental set of
cycles (together with orientation) we collect the corresponding fluxes
into the \emph{flux vector} $\vec{\alpha}\in\T^{\beta}$.  Here
$\T^\beta$ is the $\beta$-dimensional flat torus
$\R^\beta / (2\pi \Z)^\beta$.

We can thus speak of the eigenvalues of the ``operator'' $\opH^{\av}$
(which is actually an equivalence class of operators $\opH_A$).
Consider these eigenvalues
$\left\{k_n\left(\av\right)\right\}_{n=1}^\infty$ as functions of the
fluxes $\av$.  At the point $\av=0$ they are equal to the eigenvalues
of the non-magnetic operator $\opH$.  What is far from obvious is that
the behavior of $k_n(\av)$ around $\av=0$ determines the nodal surplus
of the $n$-th eigenfunction of $\opH$.



\begin{theorem}[Berkolaiko--Weyand \cite{BerWey_ptrsa14}]
  \label{thm:morse}
  Let $\Gamma$ be a quantum graph, $k_n>0$ a generic eigenvalue of
  $\opH$ on $\Gamma$ and $\sigma_{n}$ its nodal surplus. Consider the
  eigenvalue $k_n(\av)$ of the corresponding magnetic Schr\"odinger
  operator as a function of $\av$.  Then $\av=0$ is a smooth
  non-degenerate critical point of $k_{n}\left(\av\right)$ and its
  Morse index is equal to the nodal surplus $\sigma_n$.
\end{theorem}

The Morse index of a function is the number of negative eigenvalues of the
Hessian evaluated at a critical point of this function.  Because
Hessians will play a large role in our proofs, we should set up
notation carefully.

\begin{definition}
  \label{def:Hessian}
  Let $f(\vec{x},\vec{y})$ be a twice differentiable function of a
  finite number of variables $\vec{x}=(x_1,\ldots,x_n)$ and
  $\vec{y}=(y_1,\ldots,y_m)$.  The \emph{Hessian} of $f$ with respect
  to variables $\vec{y}$ evaluated at the point $\vec{x} = \vec{x}^*$
  and $\vec{y} = \vec{y}^*$ is a matrix of second derivatives
  \begin{equation}
    \label{eq:Hessian_def}
    H_{\vec{y}}(f)(\vec{x}^*, \vec{y}^*)
    := \left[ \frac{\partial^2 f }
      {\partial y_i \partial y_j}(\vec{x}^*, \vec{y}^*) \right]_{i,j=1}^m.
  \end{equation}
  For a symmetric matrix $A$ we will denote by $\M[A]$ the number of
  its negative eigenvalues (\emph{Morse index}).
\end{definition}

With this notation in hand, Theorem \ref{thm:morse} can be summarized as
\begin{equation}
  \label{eq:magnetic_nodal_formula}
  \sigma_n = \M\left[H_{\av}(k_n)\left(\vec{0}\right)\right],
\end{equation}
where $k_n = k_n\left(\av\right)$.

\subsection{Bond scattering matrix and secular equation.}
\label{sec:Reduction}

Solving the eigenvalue equation $\opH_Af = k^2f$ for $k>0$ on every edge
and applying the vertex condition we arrive at the \emph{secular
  equation} on the eigenvalues $k$
\cite{Bel_laa85,KotSmi_ap99,KotSmi_prl97}.  In this subsection we
review this procedure.

Taking, without loss of
generality, the magnetic potential to be constant on each edge, the
solution on the edge $e$ is given by
\begin{equation}
  \label{eq:amp}
  f(x_e) = a_{e}e^{-i(k+A_e)(l_e-x_e)}+a_{\hat{e}}e^{-i(k-A_e)x_e}.
\end{equation}
Thus each edge corresponds to two ``directed'' coefficients, $a_e$ and
$a_{\hat{e}}$, which can be viewed as the amplitudes\footnote{More
  precisely, the coefficients $a_e$ and $a_{\hat{e}}$ are the amplitudes
  of the waves measured just before they hit the vertex they are
  traveling to; this causes the slightly unusual form of equation
  \eqref{eq:amp} but fits with the form we choose for our secular
  equation \eqref{eq:sec_condition}.} of waves traveling in and
against the chosen direction of $e$.  The label $\hat{e}$ is called
the \emph{reversal} of the (label) $e$.  The hat can be viewed as a
permutation (fixed-point free involution) on the set of labels if we
extend it by $\hat{\hat{e}}=e$.  If we introduce variables
$x_{\hat{e}} := L-x_e$ and $A_{\hat{e}}=-A_e$, expression
\eqref{eq:amp} becomes nicely symmetric.

Let us fix a particular representative of the equivalence class of
operators $\opH_A$ (see \cite{BerWey_ptrsa14} for more detail).
Choose a set $\Cut$ of $\beta$ edges whose removal from $\E$ does not
disconnect the graph.  The remaining graph is a \emph{spanning tree}.
We set the magnetic potential to be non-zero only on the edges
$e_j\in\Cut$, with $A_{e_j} = \alpha_j / l_{e_j}$.

Assume $f$ is an eigenfunction of $\opH_A$ corresponding to the
eigenvalue $k>0$. According to \eqref{eq:amp}, $f$ is uniquely
determined by a vector of coefficients
\begin{equation}
  \label{eq:vector_amplitudes}
  \vec{a}\in\C^{2E}, \qquad
  \vec{a}=\left(a_{1},a_{\hat{1}},...,a_{E},a_{\hat{E}}\right).
\end{equation}
Note that we have chosen a specific order of the edge labels.
Imposing the vertex conditions on $f$ and simplifying the result we
arrive to the condition
\begin{equation}
  \left(\Id -
    e^{i\ba}e^{ik \Lmat}\Smat\right)\vec{a} = 0,
  \label{eq:sec_condition}
\end{equation}
where $\Lmat$ is the diagonal matrix of the edge lengths and $\ba$ is
the diagonal matrix of edge-integrated magnetic potential values,
which in our chosen representation of potential are given by 0,
$\alpha_e$ or $\alpha_{\hat{e}} = -\alpha_e$.  Assuming for the moment
that the edges in $\Cut$ are the first in the order established by
\eqref{eq:vector_amplitudes}, we have
\begin{equation*}
  \Lmat = \diag(l_1,l_1,l_2,l_2, \ldots, l_E, l_E)
  \qquad
  \ba = \diag(\alpha_1, -\alpha_1, \alpha_2, -\alpha_2, \ldots, 0, 0).
\end{equation*}
The matrix $\Smat$ (called the \emph{bond-scattering matrix}) is
unitary.  For a graph with Neumann (or Dirichlet) vertex conditions it
has constant coefficients given by the following rules.  Let $e$ be
the label corresponding to a directed edge terminating at vertex $v$
and let $\deg(v)$ denote the degree of the vertex $v$.  Then the
elements of $\Smat$ are
\begin{equation}
  \label{eq:def_S_matrix}
  S_{e',e} =
  \begin{cases}
    \frac{2}{\deg(v)} - 1, & \mbox{if }e'=\hat{e}, \\
    \frac{2}{\deg(v)}, & \mbox{if $d'$ originates at $v$ and }e'\neq\hat{e},\\
    0 & \mbox{otherwise}.
  \end{cases}.
\end{equation}
Finally, if we choose to impose Dirichlet condition at a vertex $v$
of degree 1 then the corresponding element of $\Smat$ changes from $1$
given by \eqref{eq:def_S_matrix} to $-1$.

It is an explicit computation that each term in the product
$e^{i\ba}e^{ik \Lmat}\Smat$ is a unitary matrix.  The multiplicity of
$k^2\neq 0$ as an eigenvalue in the spectrum of $\opH$ is equal to the
dimension of the kernel of $\Id-e^{i\ba}e^{ik \Lmat}\Smat$.  In other
words, the geometric multiplicity of $k^2$ is equal to the algebraic
multiplicity of $k>0$ as the root of the \emph{secular equation}
\begin{equation}
  \label{eq:secular_eq}
  \tilde{F}\left(k;\av\right) := \det\left(\Id -
    e^{i\ba}e^{ik \Lmat}\Smat\right) = 0,
\end{equation}
For more complete understanding of the scattering approach we refer
the reader to \cite{BerKuc_graphs,GnuSmi_ap06}.

\subsection{The torus flow}
\label{sec:torus_approach}

We now describe an approach pioneered by Barra and Gaspard
\cite{BarGas_jsp00} in their study of eigenvalue spacing of small
quantum graphs.

\begin{figure}
  \centering
  (a)\includegraphics[scale=0.75]{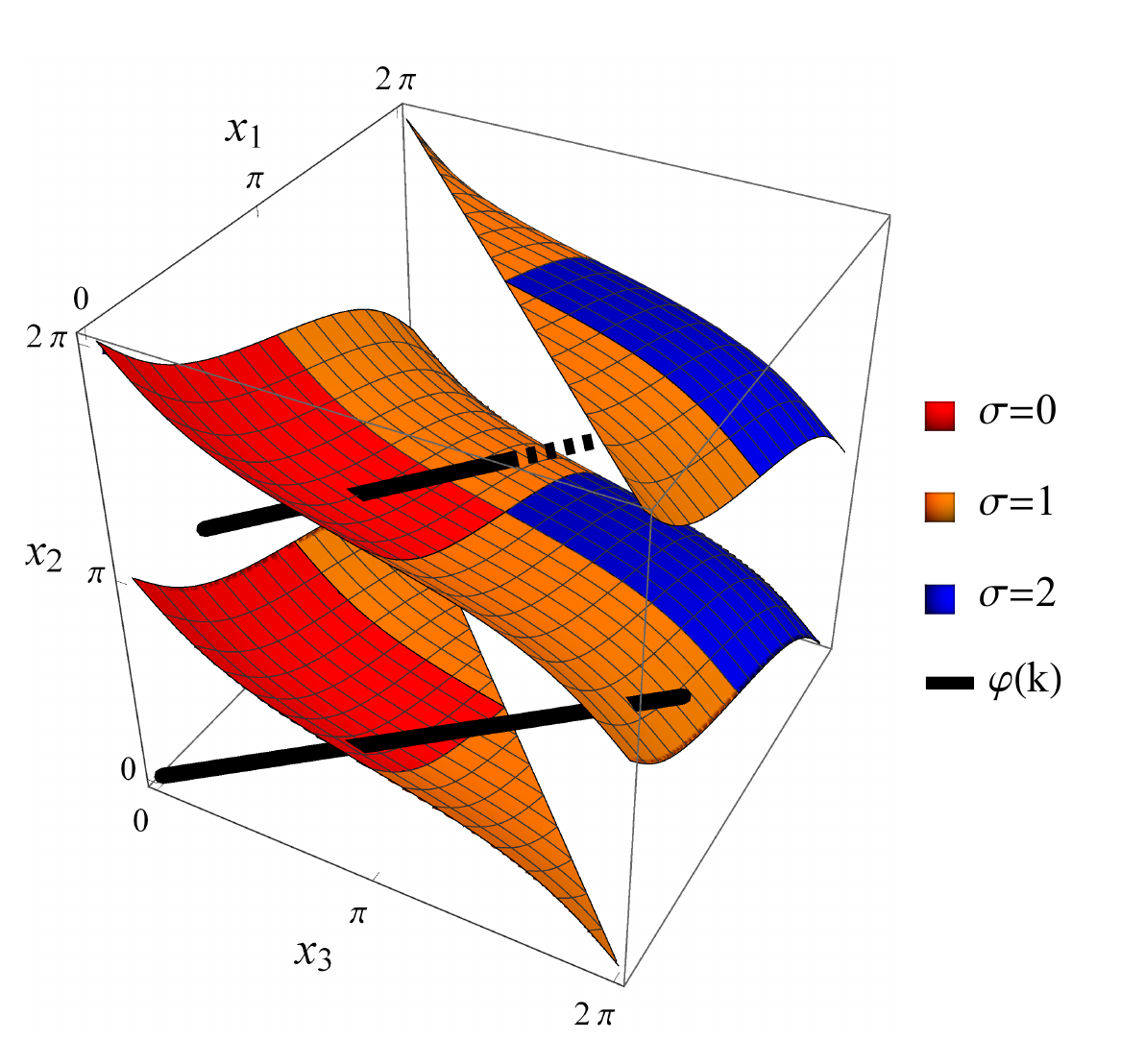}
  (b)\includegraphics[scale=1.5]{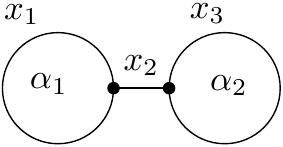}
  \caption{(a) The torus flow defined by (\ref{eq:def_flow}) hits the
    secular manifold.  The values of $k$ for which this happens are
    the eigenvalues of the graph.  The secular manifold in the figure
    is of the ``dumbbell graph'' graph as analyzed in Appendix
    \ref{sec:dumbbell}. The secular manifold is colored according to
    the values of the nodal surplus function, equation
    (\ref{eq:surplus_function_def}). (b) The ``dumbbell'' graph with
    torus coordinates marked on corresponding edges.}
  \label{fig:flow_hits_manifold}
\end{figure}

Let $\Gamma$ be a metric graph
and define
\begin{equation}
  F: \T^{E}\times\T^{\beta}\to\C,
  \qquad
  F\left(\xv;\av\right) = \det\left(\Id-e^{i\ba}e^{i\bx}\Smat\right),
  \label{eq:secfun_torus}
\end{equation}
where $\Smat$ has been calculated according to prescription
\eqref{eq:def_S_matrix} and
\begin{equation*}
  \bx = \diag\left(x_1,x_1,...x_E,x_E\right).
\end{equation*}
Consider the linear flow on the torus $\T^E$
\begin{equation}
  \label{eq:def_flow}
  \varphi:\R\rightarrow\T^{E},
  \qquad
  \varphi\left(k\right) = k\lv \mod 2\pi ,
\end{equation}
where $\lv = (l_1,\ldots,l_E)$ is the vector of lengths of $\Gamma$.
Observe that
\begin{equation}
  \tilde{F}\left(k;\av\right)
  = F\left(\varphi(k);\av\right).
\end{equation}
For the rest of the paper, $F$ will be referred to as the
\textit{secular function} of $\Gamma$.  Define the \emph{secular
  manifold},
\begin{equation}
  \S:=\left\{ \xv : F\left(\xv;0\right) = 0\right\}
  \subset\T^{E}
  \label{eq:secMan}
\end{equation}
(note that it is a slight misnomer, as generally $\S$ is an algebraic
variety with singularities and not a smooth manifold). The spectrum of
$\opH$ on $\Gamma$ can be described as the values of $k$ for which the
flow $\phi$ hits the secular manifold,
\begin{equation}
  \label{eq:hitting_times}
  \left\{k_n\right\}_{n=1}^\infty \setminus \{0\}
  = \left\{ k>0 : \varphi(k)\in\S \right\},
\end{equation}
see Figure~\ref{fig:flow_hits_manifold} for an example.  Moreover, the
multiplicity of the eigenvalue $k$ is the same as the algebraic
multiplicity of the root $\xv = \varphi(k)$ of $F(\xv;0)$.  We remark
that we took some pains to exclude zero eigenvalue from
\eqref{eq:hitting_times}.  Zero may or may not be an eigenvalue of the
graph (it is not an eigenvalue if we have some Dirichlet vertices),
and in general its multiplicity is different from the multiplicity of
$\xv=\vec{0}$ as a root of $F(\xv;0)$; this topic is studied in some
detail in \cite{FulKucWil_jpa07}.

We can similarly define $\S(\av)$ whose piercings by the flow will
give the eigenvalues $k_n\left(\av\right)$ of the magnetic operator.

A surprising consequence of Theorem~\ref{thm:morse}, pointed out in
\cite{Ban_ptrsa14}, is that one can read off the nodal surplus information
directly off the secular manifold $\S$.

\begin{theorem}
  \label{thm:surplus_function}
  For a graph $\Gamma$, define the \emph{nodal surplus function} by
  \begin{equation}
    \label{eq:surplus_function_def}
    \sigma:\S \to \{0,\ldots,\beta\} \qquad
    \sigma\left(\xv\right) := \M\left[-\frac{H_{\av}(F)\left(\xv;\vec0\right)}{\gradx
        F\cdot\lv}\right].
  \end{equation}
  Then the function $\sigma$ is independent of $\lv$ and if $k_n>0$ is
  generic, it gives the nodal surplus of the corresponding
  eigenfunction,
  \begin{equation}
    \sigma\left(k_{n}\lv\right)=\sigma_{n}.
    \label{eq:surplus_relation}
  \end{equation}
\end{theorem}

See Figure \ref{fig:flow_hits_manifold}(a) for a demonstration of the nodal surplus function for a particular graph.

\begin{remark}
  \label{rem:oracle}
  Before we sketch the proof of the Theorem, let us discuss its
  significance.  We have defined an ``oracle'' function which
  calculates the nodal surplus from $\xv$ alone.  For different
  choices of $\lv$ a given point $\xv\in\S$ may be reached by the flow
  $\varphi(k)$ at very different values of $k$, if it is reached at all.
  The corresponding eigenfunctions have very different numbers of
  zeros and come at different sequence numbers in the spectrum of
  their graph, yet the nodal surplus remains the same!
\end{remark}

\begin{proof}[Proof of Theorem~\ref{thm:surplus_function}]
  Equation~\eqref{eq:surplus_relation} follows directly from
  Theorem~\ref{thm:morse} by calculating the Hessian in terms of the function $F$,
  \begin{equation*}
    H_{\av}(k_n)\left( \vec{0} \right)
    = -\frac{H_{\av}(F)\left(\xv;\vec0\right)}{\gradx
      F\cdot\lv}.
  \end{equation*}
  This was performed in \cite{Ban_ptrsa14}, where it was also pointed
  out that $\gradx F$ has all entries of the same sign (up to an
  overall phase) or $0$.  Moreover, if $k_n$ is simple then at
  $\xv = k_n\vec{l}$ the gradient $\gradx F \neq 0$, so the function
  $\sigma$ is well defined.  Since only the sign of $\gradx F\cdot\lv$
  is important in calculating the Morse index, and $\vec{l}$ entries
  are all positive, the value of $\sigma$ remains the same whatever
  the lengths of the graph's edges are.
\end{proof}

We can now sketch out the path to proving our first main result,
Theorem~\ref{thm:existence_of_distrib}.  We will first explain that
the surplus function $\sigma$ is well defined on a large subset of
$\Sigma$ (Section~\ref{sec:Sigma_reg_generic}).  The surplus
distribution will then be represented as an integral over $\Sigma$
with an appropriate measure (Section~\ref{sec:BarraGaspard_measure}).
Finally, we will exhibit a symmetry in the function $\sigma$ which
will give us the symmetry of the surplus distribution
(Section~\ref{sec:total_symmetry}).

\subsection{Regular and generic subsets of $\Sigma$}
\label{sec:Sigma_reg_generic}

To effectively use Theorem~\ref{thm:surplus_function} we need to
understand its domain of applicability.  First, the function $\sigma$
is not defined if $\gradx F = 0$.  Second, it would be
convenient to be able to tell if $k_n$ is going to be generic just by
looking at the point $\xv$ on the torus.  This motivates us to define
and study properties of two subsets of $\Sigma$, $\Sreg$ (regular) and
$\Sgen$ (generic).

\begin{theorem}[Colin~de~Verdi\`ere \cite{CdV_ahp15}]
  \label{thm:Sreg_def_prop}
  Let $\Gamma$ be a nontrivial standard graph.  Then the set
  \begin{equation}
    \label{eq:Sigma_reg_def}
    \Sreg= \left\{\xv\in\S :
    \gradx F\left(\xv;0\right)\ne0\right\},
  \end{equation}
  has the following properties.
  \begin{enumerate}
  \item The algebraic variety $\S \setminus \Sreg$ is of co-dimension
    at least one in $\S$, which in turn has co-dimension one in $\T^E$.
  \item $\Sreg$ is an open manifold (possibly disconnected)
    with the normal at a point $\xv$ given by
    \begin{equation}
      \label{eq:normalS}
      \normal_j(\xv) = C\left(\bigl|a_j\bigl|^2
        + \bigl|a_{\hat{j}}\bigl|^2\right),
    \end{equation}
    where $C$ is a normalization constant and $\vec{a}$ is the
    eigenvector of the eigenvalue 1 of the matrix
    $e^{i\bx}\Smat$.
  \item \label{item:eigenvalue1}
    $k>0$ is a simple eigenvalue of the
    graph $\Gamma$ if and only if $\varphi(k) \in \Sreg$.
    Equivalently, $\xv\in\Sreg$ if and only if $k=1$ is a
    simple eigenvalue of the graph $\Gamma$ with the lengths
    $\vec{l}=\xv$; the corresponding eigenfunction will be called the
    \emph{canonical eigenfunction}.
  \end{enumerate}
\end{theorem}

\begin{remark}
  We would like to comment here on one important aspect of the proof
  of Theorem~\ref{thm:Sreg_def_prop}.
  Since we defined $\S$ as the zero set of a complex function, one
  would expect the co-dimension to be 2.  It is 1 because the function
  is actually real up to a smooth phase factor.  More precisely, the function
  \begin{equation}
    F_{R}\left(\xv;\av\right) :=
    \frac{e^{-i(x_1+\ldots+x_E)}}{\sqrt{\det(\Smat)}} F\left(\xv;\av\right)
    \label{eq:Real_sec_func}
  \end{equation}
  is real and share the same zero set as $F$ \cite{KotSmi_ap99}.  The
  degenerate cases $F_R\equiv0$ and $\Sreg=\emptyset$ can be excluded
  because the spectrum $k_n$ is discrete, continuous in $\lv$ and
  generically simple (on nontrivial graphs).
\end{remark}

\begin{remark}
  \label{rem:FR_surplus function}
  The real version of $F$ we defined in \eqref{eq:Real_sec_func} can
  be used in place of $F$ in the definition of the surplus function on
  $\Sigma$,
  \begin{equation}
    \label{eq:surplus_function_def_alt}
    \sigma\left(\xv\right) = \M\left[-\frac{H_{\av}(F_R)\left(\xv;\vec0\right)}{\gradx
        F_R\cdot\lv}\right].
  \end{equation}
  Indeed,
  \begin{align*}
    H_{\av}(F_R)
    &= \frac{e^{-i(x_1+\ldots+x_E)}}{\sqrt{\det(\Smat)}}
    H_{\av}(F), \\
    \gradx F_R
    &= \frac{e^{-i(x_1+\ldots+x_E)}}{\sqrt{\det(\Smat)}}
    \gradx F
    + \frac{\gradx e^{-i(x_1+\ldots+x_E)}}{\sqrt{\det(\Smat)}} F
    = \frac{e^{-i(x_1+\ldots+x_E)}}{\sqrt{\det(\Smat)}}
    \gradx F,
  \end{align*}
  where the second term in $\gradx F_R$ disappears since on $\S$ we have
  $F(\xv;0)=0$.  The prefactors then cancel when taking the quotient
  in \eqref{eq:surplus_function_def_alt}.
\end{remark}

To have a well defined nodal count (and to apply
Theorem~\ref{thm:surplus_function}) we need $k_n$ to be generic, a
quality that is determined by looking at the corresponding
eigenfunction.  We would like to be able to determine it by looking
directly at the secular manifold $\Sigma$.  An observant reader would
protest that there is no such thing as the ``corresponding
eigenfunction'' at $\xv\in\S$: the eigenfunction depends on the choice
of lengths $\vec{l}$, as was pointed out in Remark~\ref{rem:oracle}.
However, it turns out that all eigenfunctions that arise\footnote{that
  is eigenfunctions $f_n$ corresponding to an eigenvalue $k_n$ of a
  graph $\Gamma$ with lengths $\vec{l}$ such that
  $k_n \vec{l} = \xv \mod 2\pi$} from a given $\xv\in\Sreg$ share many
properties, such as the values they take on the vertices.  In
particular one can just check the genericity of the canonical
eigenfunction, defined in
Theorem~\ref{thm:Sreg_def_prop}(\ref{item:eigenvalue1}).

The generic eigenvalues need to be simple, therefore we are looking
for a subset of $\Sreg$ (see Theorem \ref{thm:Sreg_def_prop}(\ref{item:eigenvalue1})).
Next we need to exclude the points on $\Sigma$ where the corresponding
eigenfunctions vanish on a vertex.  The point $\xv \in \Sreg$
uniquely determines the one-dimensional null space of
$\Id-e^{i\bx}\Smat$.  If $\vec{a}$ is a vector spanning this null
space, it is proportional to any non-zero column of the adjugate of
$\Id-e^{i\bx}\Smat$ and therefore each entry of $\vec{a}$ is a
trigonometric polynomial in $\xv$.  From \eqref{eq:amp}, the value of
any eigenfunction at a vertex $v$ is given by
\begin{equation*}
  f(v) =  e^{-i kl_e}a_e + a_{\hat{e}},
\end{equation*}
where $e$ is any vector coming out of $v$ and $kl_e$ is the $e$-th
component of the point $\xv$ on the torus.  Defining
\begin{equation}
  \label{eq:zero_Sigma}
  \Sigma_0 := \left\{ \xv\in\Sreg :
    \prod_{e} \left(e^{-i \tx_e}a_e + a_{\hat{e}} \right) = 0 \right\},
\end{equation}
we then have the following theorem.

\begin{theorem}
  \label{thm:Sigma_generic}
  If $\Gamma$ is a nontrivial standard graph, the set
  \begin{equation}
    \label{eq:Sigma_generic_defn}
    \Sgen = \Sreg \setminus \Sigma_0
  \end{equation}
  is a non-empty submanifold of $\T^E$ of co-dimension 1.  An
  eigenvalue $k_n$ of $\Gamma$ with lengths $\vec{l}$ is generic if
  and only if
  \begin{equation*}
    \varphi(k_n) \in \Sgen,
  \end{equation*}
  where $\varphi(k) = k\lv \mod 2\pi$.

  The surplus function $\sigma$ is constant on every connected
  component of the manifold $\Sgen$.
\end{theorem}

\begin{proof}
  The set $\Sgen$ is the set consisting of all the generic eigenvalues
  by its construction.  To show that it is non-empty we use the
  results of \cite{BerLiu_jmaa17}: for a typical choice of lengths
  $\lv$, every eigenvalue is either generic or its eigenfunction is
  supported on a single loop.  But for graph which is not a cycle, the
  proportion of the loop eigenstates in the spectrum is
  $\L_{loops} / 2\L \leq 1/2$, where $\L_{loops}$ is the total length
  of all loops and $\L$ is the total length of all edges of the graph
  (including loops).  This result easily follows from the Weyl
  estimate for the number of eigenvalues combined with our explicit
  knowledge of the loop eigenvalues (see
  Appendix~\ref{sec:volume_Sgen}).  Since the set $\Sigma_0$ is a
  compact subset of $\Sreg$, the set $\Sreg \setminus \Sigma_0$ is a
  submanifold of $\Sreg$ of the same dimension.

  Finally, the surplus function is constant on every connected
  component of $\Sgen$ because the eigenvalues of the matrix in the
  definition of $\sigma$ vary continuously with $\xv$.  To change the
  Morse index, one of them has to become zero or $F_R\cdot\lv$ has to
  vanish, both of which are impossible on generic eigenvalue, by
  Theorems~\ref{thm:morse} and \ref{thm:Sreg_def_prop}.
\end{proof}

\subsection{Ergodicity and the Barra-Gaspard measure}
\label{sec:BarraGaspard_measure}

The main idea of Barra and Gaspard \cite{BarGas_jsp00} was that if one
wants to calculate the average of a certain function of the spectrum
of a quantum graph, it is often possible to redefine this function in
terms of the $\xv$ torus coordinates instead and then integrate over
the secular manifold $\Sigma$ with an appropriate measure.  This idea
was applied to eigenvalue statistics in the original paper
\cite{BarGas_jsp00}, used to study eigenfunction statistics
\cite{BerWin_tams10}, eigenfunction scarring \cite{CdV_ahp15},
band-gap statistics of periodic structures
\cite{BanBer_prl13,GalBan_jmps16,ExnTur_prep17} and statistics of
topological resonances \cite{ColTru_prep16}.

\begin{definition}[Barra-Gaspard measure \cite{BarGas_jsp00,CdV_ahp15}]
  Let $\Gamma$ be a quantum graph with lengths $\lv$.  The
  \emph{Barra-Gaspard} measure on the smooth manifold $\Sgen$ is the
  lengths dependent probability measure
  \begin{equation}
    \label{eq:BG_measure_def}
    d\bgm:=\frac{\left|\normal(\xv)\cdot\lv\right|}{C}d\sigma,
  \end{equation}
  where $\normal$ is the unit vector field normal to $\Sgen$,
  $d\sigma$ is the surface element of $\Sgen$ induced by the Euclidean
  metric and $C=\int_{\Sgen}\left|\normal\cdot\lv\right|d\sigma$ is
  the normalization constant which depends on the lengths $\lv$.
\end{definition}

\begin{theorem}[Barra--Gaspard \cite{BarGas_jsp00}, Berkolaiko-Winn
  \cite{BerWin_tams10}, Colin~de~Verdi\'ere \cite{CdV_ahp15}]
  \label{thm:BG_measure_ergodicity}
  Let $\Gamma$ be a nontrivial standard graph. Then
  $\bgm$ satisfies the following properties:
  \begin{enumerate}
  \item It is a Radon measure on $\Sgen$.
  \item \label{item:ergodic} If the lengths $\lv$ are rationally
    independent, then for any \emph{Riemann integrable} function
    $f:\Sgen\rightarrow\R$
      \begin{equation}
        \label{eq:ergodic_thm_torus}
        \lim_{N\to\infty} \frac{1}{\left|\genN\right|}
        \sum_{n\in\genN} f\left(\varphi(k_n)\right)
        = \int_{\Sgen} fd\bgm,
      \end{equation}
      where $\varphi(k) = k\lv \mod 2\pi$ and $\genN$ is the set of
      indices $1 \leq n \leq N$ such that $k_n$ is generic.
  \end{enumerate}
\end{theorem}

\begin{remark}
  In \cite{BarGas_jsp00,CdV_ahp15} this was proven for continuous
  functions for a measure defined on $\Sreg$ instead.  Restricting it
  to $\Sgen$ does not change any substance.  The adjustment in the
  normalizing constant is shown in Appendix~\ref{sec:volume_Sgen} to
  be
  \begin{equation*}
    \frac{C^g}{C^{reg}} = 1 - \frac{\L_{loops}}{2\L}.
  \end{equation*}
  Extending the result from continuous to Riemann integrable functions
  is done using proposition 4.4 of \cite{BerWin_tams10}.

  Note that part (\ref{item:ergodic}) of the theorem cannot be
  extended to include all measurable functions since the set of our
  sample points has measure zero.  A Birkhoff-type result holding for almost
  every starting point would not be sufficient for us since our flow
  piercing $\Sigma$ (equation \eqref{eq:def_flow}) has the fixed
  starting point $\varphi(0)=\vec0$.
\end{remark}

\subsection{The surplus function $\sigma$ is even}
\label{sec:total_symmetry}

We now exhibit a symmetry in $\S$ that has a profound effect on
the nodal surplus distribution.

\begin{lemma}
  \label{lem:total_symmetry}
  Let $\Gamma$ be a nontrivial standard graph with lengths $\lv$.  The inversion
  \begin{equation}
    \label{eq:inversion_def}
    \inv:\T^{E}\rightarrow\T^{E},\qquad
    \inv\left(\xv\right)= -\xv = 2\pi - \xv \mod 2\pi
  \end{equation}
  is a measure preserving homeomorphism of
  $\left(\Sgen,\bgm\right)$ to itself.

  Furthermore, under the mapping $\inv$, the surplus function
  transforms as
  \begin{equation}
    \label{eq:antisymmetric_sigma}
    \sigma\circ\inv(\xv) = \sigma\left(-\xv\right)
    = \beta-\sigma\left(\xv\right).
  \end{equation}
\end{lemma}

An example of the symmetry~(\ref{eq:antisymmetric_sigma}) can be observed
in Figure~\ref{fig:flow_hits_manifold}(a).

\begin{proof}
  In order to prove that $\inv$ is a measure preserving homeomorphism
  of $\left(\Sgen,\bgm\right)$ to itself, first observe that $\inv$ is
  smooth, has a smooth inverse (itself), and has Jacobian determinant
  equal to 1 in absolute value.  We are only left to show that
  \begin{equation}
    \label{eq:to_establish}
    \xv\in\Sgen \Rightarrow \inv(\xv)\in\Sgen,
    \qquad\mbox{and}\qquad   \left| \hat{n}(\xv) \cdot \lv \right| =
      \left| \hat{n}(\inv(\xv)) \cdot \lv \right|.
  \end{equation}

  Let $\xv \in \Sgen$ and let $f$ be the eigenfunction of the simple
  eigenvalue 1 guaranteed by
  Theorem~\ref{thm:Sreg_def_prop}(\ref{item:eigenvalue1}).  On the
  edge $e$ the function $f$ has the form
  \begin{equation}
    \label{eq:function_form}
    f_e(x) = C_e \cos(x-\theta_e),
  \end{equation}
  for some $\theta_e \in [0,2\pi)$ and $C_e>0$.  We remark that
  $C_e^2$ is equal to the $e$-th component of the normal vector
  $\normal(\xv)$, namely
  $\bigl|a_j\bigl|^2 + \bigl|a_{\hat{j}}\bigl|^2$.  The function $f_e$
  is analytic and $2\pi$-periodic; we can view it as being defined by
  (\ref{eq:function_form}) not just on the edge $[0,l_e]$ but on the
  whole real line.

  We now let $\tilde{f}_e(x) = f_e(-x)$.  We claim it is an
  eigenfunction with eigenvalue $1$ of the graph with lengths
  $\tilde{l}_e = 2\pi - l_e$.  Indeed, it obviously solves the
  eigenvalue equation on every edge and satisfies the vertex
  conditions since
  \begin{align*}
    \tilde{f}_e(0) & = f_e(0),
    & \tilde{f}_e\left(\tilde{l}_e\right)
    &= f_e(l_e - 2\pi) = f_e(l_e), \\
    \tilde{f}_e'(0) &= -f_e'(0),
    & \tilde{f}_e'\left(\tilde{l}_e\right)
    & = -f_e'(l_e).
  \end{align*}
  This construction is obviously invertible so the multiplicity of
  eigenvalue 1 at $\xv$ and at $\inv(\xv)$ are the same.  Similarly,
  $\tilde{f}$ is generic if and only if $f$ is and
  the first part of (\ref{eq:to_establish}) is established.

  The normal vectors at the two points coincide:
  $\normal(\xv) = \normal(\inv(\xv))$ up to a sign because what
  appears in (\ref{eq:normalS}) is the square of the amplitude $C_e$
  of the cosine.  Therefore, the transformation $\inv$ is measure preserving.

  Conjugating $F_R$,
  \begin{equation*}
    F_R\left(\xv;\av\right) =
    \frac{e^{-i(x_1+\ldots+x_E)}}{\sqrt{\det(\Smat)}}
    \det\left(\Id-e^{i\ba+i\bx}\Smat\right),
  \end{equation*}
  which we know to be real, we get
  \begin{equation}
    \label{eq:F_real_is_even}
    F_R\left(\xv;\av\right) = \cc{F_R\left(\xv;\av\right)}
    = \pm F_R\left(-\xv;-\av\right),
  \end{equation}
  where the sign depends on whether $\det\Smat$ is equal to $1$ or $-1$.
  Therefore
  \begin{equation}
    \label{eq:Hessian_gradient_under_inversion}
    H_{\av}(F_R)\left(-\xv;\vec{0}\right)
    = \pm H_{\av}(F_R)\left(\xv;\vec{0}\right),
    \quad\mbox{while}\quad
    \gradx F_R\left(-\xv;\vec{0}\right)
    = \mp \gradx F_R\left(\xv;\vec{0}\right).
  \end{equation}
  If $A$ is a nondegenerate symmetric $\beta\times \beta$ matrix, we obviously
  have $\M[-A] = \beta-\M[A]$ and thus
  equation~\eqref{eq:antisymmetric_sigma} follows from
  Remark~\ref{rem:FR_surplus function}.
\end{proof}

\subsection{Proof of Theorem~\ref{thm:existence_of_distrib}}
\label{sec:proof_of_existence_and_symmetry}

We collect all the preceding discussion together for the proof of our
first main theorem.

\begin{proof}[Proof of Theorem~\ref{thm:existence_of_distrib}]
  By Theorem~\ref{thm:Sigma_generic} the surplus function is constant
  on each connected component of $\Sgen$, so it is actually
  continuous. The frequency $\Pr{s}$ (see equation
  \eqref{eq:surplus_distr_defn}) can be obtained from
  Theorem~\ref{thm:BG_measure_ergodicity} by setting $f$ to be the
  indicator function of the set $\sigma^{-1}(s)$,
  \begin{equation}
    \label{eq:frequency_from_measure}
    \Pr{s} = \bgm\left( \sigma^{-1}(s) \right) = \mu\left( \sigma^{-1}(s) \right).
  \end{equation}

  Abbreviating $\bgm$ to $\mu$ to avoid clutter, we use the properties
  of $\inv$ as seen in Lemma~\ref{lem:total_symmetry},
  \begin{equation*}
    \mu\left( \sigma^{-1}(s) \right)
    =\mu\left( \inv^{-1} \circ \sigma^{-1}(s) \right)
    =\mu\left( (\sigma\circ\inv)^{-1}(s) \right)
    =\mu\left( \sigma^{-1}(\beta-s) \right),
  \end{equation*}
  which proves that
  \begin{equation}
    \Pr{s} = \Pr{\beta-s}.
  \end{equation}
\end{proof}

\section{Nodal surplus of graphs with block structure}
\label{sec:graphs_block_structure}

The aim of this section is the proof of
Theorem~\ref{thm:binomial_distr}.  After introducing some additional
tools (Sections~\ref{sec:scattering_def} and \ref{sec:splitting}) and
setting up the definitions (Section~\ref{sec:block_defs}) we will see
in Section~\ref{sec:block_structure} that the nodal surplus function
can be localized to a block of the graph (see
Fig.~\ref{fig:vertex_separation}).  After studying properties of the
local surplus functions in Sections~\ref{sec:surplus_determined} and
\ref{sec:bridge_symmetry} we get a handle on their probability
distributions in Section~\ref{sec:proof_main2} and hence prove our
second main result, Theorem~\ref{thm:binomial_distr}.

Section~\ref{sec:scattering_def} contains a review of well-known facts
that we need in the proofs of subsequent sections; a reader not
interested in the details of the proofs may skip it entirely.
Section~\ref{sec:splitting} is also needed only for the subsequent
proofs (and only in its simplest form).  However it contains a new
formulation of a well-known idea which may turn out to be a useful
in other settings.

\subsection{Scattering from a graph}
\label{sec:scattering_def}

One can probe spectral properties of a graph by attaching (infinite)
leads to it and considering the scattering of plane waves coming in
from infinity
\cite{GerPav_tmp88,KotSmi_ap99,KosSch_jpa99,KotSmi_prl00,DavPus_apde11,DavExnLip_jpa10,BanBerSmi_ahp12}.

Let $\Gamma=\left(\E,\V\right)$ be a standard graph and let
$\tilde{\Gamma}$ be the non-compact quantum graph constructed by
attaching $M$ infinitely long edges (leads) to some vertices of the
graph and imposing Neumann vertex conditions there.

A solution $f$ of the eigenvalue equation $\opH_Af = k^2f$ on
$\tilde{\Gamma}$ with $k>0$ can be described by its compact graph
coefficients $\vec{a}\in\C^{2E}$ (see (\ref{eq:amp})) together with
similar coefficients on the $j$-th infinite leads number, $c_{j,in}$
and $c_{j,out}\in\C$,
\begin{equation}
  \label{eq:lead_solution}
  f_j\left(y\right)=c_{j,in}e^{-iky}+c_{j,out}e^{iky},
\end{equation}
where $y\in[0,\infty)$ is the coordinate along the lead starting from
0 at the attachment point.  Note that $f$ is usually not an
eigenfunction since it has an infinite $L^2$ norm unless all
coefficients $c$ are zero.

Let $\vec{c}_{in}$ and $\vec{c}_{out}$ be the vectors of the
corresponding coefficients on the leads.  Imposing vertex conditions
on the vertices of the graph results in a condition similar to
\eqref{eq:sec_condition},
\begin{equation}
  \label{eq:scattering problem}
  \begin{pmatrix}
    \vec{c}_{out}\\
    \vec{a}
  \end{pmatrix}
  =
  \begin{pmatrix}
    \Id & 0\\
    0 & e^{ik\lv + i\av}
  \end{pmatrix}
  \begin{pmatrix}
    r & t'\\
    t & \tS
  \end{pmatrix}
  \begin{pmatrix}
    \vec{c}_{in}\\
    \vec{a}
  \end{pmatrix},
\end{equation}
where the entries of the subblocks $r$, $t$, $t'$ and $\tS$ are
calculated according to formula \eqref{eq:def_S_matrix}.  The
$M\times M$ matrix $r$ describes the reflection of the waves from the
attachment vertices directly back into the leads (without getting into
the compact part $\Gamma$) and is symmetric, $r^T=r$.  The matrix $t$
describes scattering of the waves from the leads into $\Gamma$, $t'$
describes scattering from $\Gamma$ into the leads and $\tS$ describes
wave scattering between edges of $\Gamma$.  In our
setting, we have $t' = \left(Jt\right)^T$, where $J$ switches around
the directed labels of the edges of the graph $\Gamma$, namely
$\left(J\av\right)_e = a_{\hat{e}}$.  We need it
because the wave traveling on $e\in\Gamma$ and scattering into the
lead travels in the opposite direction from the wave in this edge
which came from the lead.  This relation is a manifestation of the
time-reversal symmetry of the problem.

Eliminating $\vec{a}$ from equation \eqref{eq:scattering problem} and
solving of $\vec{c}_{out}$ in terms of $\vec{c}_{in}$ we obtain
\emph{scattering matrix} $Z$
\begin{equation}
  \label{eq:scattering_matrix_def}
  \vec{c}_{out} = Z \vec{c}_{in},
  \quad \mbox{with}\quad
  Z := r + t' \left(\Id - e^{ik\lv + i\av}\tS\right)^{-1}
    e^{ik\lv + i\av}t,
\end{equation}
which is well defined as long as $\Id - e^{ik\lv + i\av}\tS$ is
non-singular\footnote{It is actually shown in
  \cite{KosSch_jpa99,BanBerSmi_ahp12} that for real $k$,
  $\det\left(\Id - e^{ik\lv + i\av}\tS\right)=0$ produces removable
  poles only.}.  We remark that expanding the inverse in geometric
series results in a nice interpretation of the scattering matrix as a
summation over all paths from one lead to another weighted with their
scattering amplitudes.

The following Theorem is an amalgamation of several results appearing
in \cite{KosSch_jpa99,BanBerSmi_ahp12} and also some new results.
We consider the matrix $Z$ as a function on the torus by replacing
$k\lv$ with $\xv$.


\begin{theorem}
  \label{thm:scattering_properties}
  Consider the scattering matrix of a nontrivial standard graph
  $\Gamma$ as a function of torus coordinates $\xv$ and magnetic
  fluxes $\av$
  \begin{align}
    \label{eq:scattering_torus}
    Z(\xv,\av) &:= r + t' \left(\Id - e^{i\xv + i\av}\tS\right)^{-1}
    e^{i\xv + i\av}t \\
    \nonumber
    &= r + t' \left(e^{-i\xv - i\av} - \tS\right)^{-1} t.
  \end{align}
  It has the following properties.
  \begin{enumerate}
  \item \label{item:scat_deg_points} For a given $\av$ denote the set of
    $\xv$ such that
    \begin{equation*}
      \det\left(\Id - e^{i\xv + i\av}\tS\right) = 0,
    \end{equation*}
    by $W(\av)$.  Then, if $\xv \in W(\av)$, the compact
    graph $\Gamma$ with lengths $\lv=\xv$ and magnetic fluxes $\av$
    has an eigenfunction with eigenvalue 1 \emph{vanishing} at all
    lead attachment vertices as well as satisfying Neumann conditions
    there.
  \item \label{item:scat_unitary} For every point
    $p=\left(\xv;\av\right)\in\T^{E}\times\T^{\beta}$ such that
    $\xv \not\in W(\av)$, the matrix $Z(\xv,\av)$ is unitary.
  \item \label{item:scat_symmetry} $Z$ satisfies
    \begin{align}
      \label{eq:scat_conjugation}
      Z(-\xv,-\av) &= \cc{Z(\xv,\av)}, \\
      \label{eq:scat_mag_reversal}
      Z(\xv,-\av) &= Z(\xv,\av)^T,
    \end{align}
  \end{enumerate}
\end{theorem}

\begin{proof}
  Parts (\ref{item:scat_deg_points}) and (\ref{item:scat_unitary}) are
  well known, see, for example, \cite[Thm~3.1 and Thm~3.3]{KosSch_jpa99}
  or \cite[Lem~2.3 and Thm~2.1(2)]{BanBerSmi_ahp12}.

  Equation~\eqref{eq:scat_conjugation} follows by conjugating the
  definition of $Z$ and using the fact that $r$, $t$, $t'$ and $\tS$
  have real entries.

  Equation~\eqref{eq:scat_mag_reversal} can be derived from \cite[Cor
  3.2]{KosSch_jpa99}, but we prefer to give a direct proof,
  introducing a useful technique.  Let $J$ denote the permutation
  matrix switching the orientation of the edge labels.  The matrix $J$
  is an orthogonal involution, i.e.\ $J^{-1}=J^T=J$.  We observe that
  \begin{equation*}
    J\tS J = \tS^T, \qquad
    J e^{i\xv + i\av} J = e^{i\xv - i\av}
    = \left(e^{i\xv - i\av}\right)^T,
    \qquad
    t' = (J t)^T.
  \end{equation*}
  Substitute the latter equation in the form $t' = (t)^T J$ into the
  definition of $Z$ and propagate $J$ through the product,
  \begin{align*}
    Z(\xv,\av)
    &= r + t' \left(e^{-i\xv - i\av} - \tS\right)^{-1} t
    = r + (t)^T J \left(e^{-i\xv - i\av} - \tS\right)^{-1} t \\
    &= r + (t)^T \left(e^{-i\xv + i\av} - \tS^T\right)^{-1} J t
      = r + (t)^T \left(e^{-i\xv + i\av} - \tS^T\right)^{-1} (t')^T \\
    &= \left(
      r + t' \left(e^{-i\xv + i\av} - \tS \right)^{-1} t
      \right)^T
      = Z(\xv; -\av)^T.
  \end{align*}
\end{proof}

\subsection{Secular equation via a splitting of a graph}
\label{sec:splitting}

The following definition introduces the concept of a \emph{splitting}
of a graph, illustrated in Figure~\ref{fig:splitting_example}.  It
uses the notion of an edge of length zero.  This is to be viewed as
the result of contracting an edge, when the two end vertices of the
edge are merged and the Neumann conditions are imposed at the newly
formed vertex.  For further information, see \cite[Appendix
A]{BanLev_prep16} which studies convergence of the eigenvalues in the
process of contraction and Appendix~\ref{sec:zero_length_edge} which
considers the effect of setting $\tx_e$ to 0 in the secular function.
Stronger forms of convergence in more general settings are established
in the forthcoming work \cite{BerLatSuk_prep17}.

\begin{figure}
  \centering
  \includegraphics[scale=1]{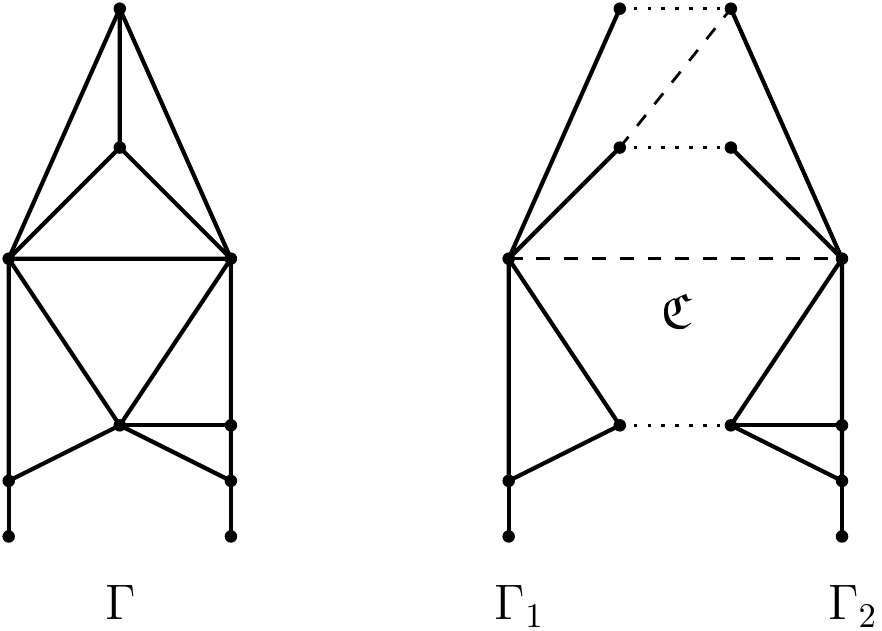}
  \caption{An example of a splitting of a graph.  The original graph
    $\Gamma$ is on the left.  The two subgraphs (in solid lines) and
    the connector set $\Cnct$ (in dashed and dotted lines) is on the
    right.  The dotted lines indicate edges of length 0.}
  \label{fig:splitting_example}
\end{figure}

\begin{definition}
  \label{def:splitting}
  A \emph{splitting} of a graph $\Gamma$ is a triple
  $[\Gamma_1,\Cnct,\Gamma_2]$, where $\Gamma_1$ and $\Gamma_2$ are two
  subgraphs of $\Gamma$, and the \emph{connector set} $\Cnct$ is a set
  of edges of $\Gamma$ with one endpoint marked as first and the other
  as second; it may also contain a number of edges of zero length
  whose endpoints coincide.  Furthermore, the following conditions are
  satisfied:
  \begin{enumerate}
  \item the first endpoints belong to the subgraph $\Gamma_1$, the
    second belong to $\Gamma_2$,
  \item if we glue edges from $\Cnct$ to subgraph $\Gamma_1$ by their
    first endpoint and to subgraph $\Gamma_2$ by their second (and
    contract the edges of zero length), we recover the graph $\Gamma$.
  \end{enumerate}
\end{definition}

For any splitting, there is a decomposition of the secular function in
terms of the scattering matrices of the subgraphs.  This is a version
of the interior-exterior duality which is nicely summarized in the
introduction to \cite{Smi_jpa09} (original articles include
\cite{SchSmi_csf95,RouSmi_jpa95}).  Similar questions on quantum
graphs were also considered in \cite{KosSch_jmp01} and
\cite[Sec.~3.3]{GnuSmi_ap06}, but we work in a more restricted
setting and the result is more compact and transparent.

\begin{theorem}
  \label{thm:splitting_decomposition}
  Let the graph $\Gamma$ have a splitting $[\Gamma_1,\Cnct,\Gamma_2]$.
  Denote by $\xv_1$, $\xv_0$ and $\xv_2$ the torus variables
  corresponding to the edges in $\Gamma_1$, $\Cnct$ and $\Gamma_2$
  correspondingly (with $\tx_e=0$ for the zero-length edges from
  $\Cnct$).  Let $\av_1$, $\av_0$ and $\av_2$ be the corresponding
  flux variables with $\av_0$ oriented in the direction from
  $\Gamma_1$ to $\Gamma_2$ .

  Attach the edges $\Cnct$ to $\Gamma_1$ by their first endpoints and
  let $Z_1(\xv_1,\av_1)$ be the $|\Cnct|\times|\Cnct|$ scattering
  matrix of $\Gamma_1$ with the edges from $\Cnct$ acting as leads.
  Define the matrix $Z_2(\xv_2,\av_2)$ analogously, and let
  $e^{i\xv_0 + i\av_0}$ be the $|\Cnct|\times|\Cnct|$ diagonal matrix
  of exponentials of the variables corresponding to $\Cnct$.  Then
  \begin{equation}
    \label{eq:splitting_decomposition}
    F(\xv; \av) = c D_1 D_2
    \det\Big(\Id - e^{i\xv_0+i\av_0}\, Z_1\,
    e^{i\xv_0-i\av_0}\, Z_2\Big),
  \end{equation}
  where $c$ is a constant, the factors $D_j = D_j(\xv_j;\av_j)$ are
  given by
  \begin{equation}
    \label{eq:D_12_def}
    D_j(\xv_j;\av_j) := \det\left(\Id - e^{i\xv_j +
        i\av_j}\Smat_j\right),
    \qquad j=1,2,
  \end{equation}
  and $\Smat_j$ is the submatrix of the bond scattering matrix $\Smat$
  of $\Gamma$ responsible for scattering from and into the edges of
  the subgraph $\Gamma_j$.  In particular, the prefactors
  $D_1$ and $D_2$ are non-zero when $\xv\in\Sgen$ and $\av=\vec{0}$.
\end{theorem}

\begin{remark}

  The determinant in equation~\eqref{eq:splitting_decomposition} has
  an elegant interpretation.  Reading the matrices right to left, we
  have the wave scattering from the subgraph $\Gamma_2$, acquiring a
  phase by traversing $\Cnct$ from $\Gamma_2$ to $\Gamma_1$,
  scattering off the subgraph $\Gamma_1$ and traversing $\Cnct$ in the
  opposite direction.  The secular function is zero when the wave
  dynamics is stationary, i.e.\ there is an eigenvector of this 4-step
  scattering process with eigenvalue one.
\end{remark}

We will use Theorem~\ref{thm:splitting_decomposition} only in its
simplest setting, when the connector set $\Cnct$ consists of one edge
of zero length.  The proof of the more general setting given above
is deferred to Appendix~\ref{sec:proof_splitting}.

\subsection{Block decomposition of a graph}
\label{sec:block_defs}

To delve deeper into the dependence of the nodal surplus on the
structure of the graph, we introduce some terminology.  We use Tutte's
definition of the graph \cite{Tutte_graph_theory} as a set of vertices
$\V$, a set of edges $\E$ and the incidence map from edges to pairs of
vertices (endpoints of the edge).  This allows for multiple edges
connecting a pair of vertices and for loop edges (if the endpoints
coincide).  A subgraph is comprised of a subset $\V_s \subseteq \V$ and a subset
of $\E_s \subseteq \E$ which form a valid graph: all endpoints of
edges in $\E_s$ are included in $\V_s$.  An intersection or union of
two subgraphs is formed by taking the intersection or union,
respectively, of both the vertex and edge sets.

\begin{definition}
  \label{def:vertex_separation}
  A \emph{vertex separation} of a graph $\Gamma$ is an ordered
  sequence of connected subgraphs $[\Gamma_1,\ldots,\Gamma_n]$ such that
  \begin{enumerate}
  \item for each $j=2,\ldots,n$, the subgraph $\Gamma_j$ has exactly one
    vertex in common with the union of all previous subgraphs,
    $\Gamma_1 \cup \cdots \cup \Gamma_{j-1}$,
  \item the union of all subgraphs is the graph $\Gamma$,
    \begin{equation*}
      \Gamma = \Gamma_1 \cup \cdots \cup \Gamma_n.
    \end{equation*}
  \end{enumerate}
  Each subgraph $\Gamma_j$ in a vertex separation we will call a
  \emph{vertex-separated block}\footnote{Our term ``vertex-separated
    block'' is a generalization of a more standard term \emph{block},
    a maximal connected subgraph without a cutvertex \cite[Section
    3.1]{Diestel_graph_theory}.  Our vertex-separated block is a
    connected union of such blocks.}.
\end{definition}

\begin{definition}
  \label{def:edge_separation}
  An \emph{edge separation} of a graph $\Gamma$ is a vertex separation
  $[\Gamma_1,\ldots,\Gamma_n]$ such that for all $j=2\ldots,n$ the
  common vertex of $\Gamma_j$ with
  $\Gamma_1 \cup \cdots \cup \Gamma_{j-1}$ is a vertex of degree one
  in $\Gamma_j$.  The edge of $\Gamma_j$ incident to this vertex we
  call the \emph{bridge of $\Gamma_j$}.

  Each subgraph $\Gamma_j$ of an edge separation we call an
  \emph{edge-separated block}.
\end{definition}

The notion of vertex separation is a generalization of the
1-separation of Tutte \cite[Section~III.1]{Tutte_graph_theory}.
Examples of vertex and edge separations are given in
Figures~\ref{fig:vertex_separation} and \ref{fig:edge_separation}.

\begin{figure}
  \centering
  \includegraphics[scale=1]{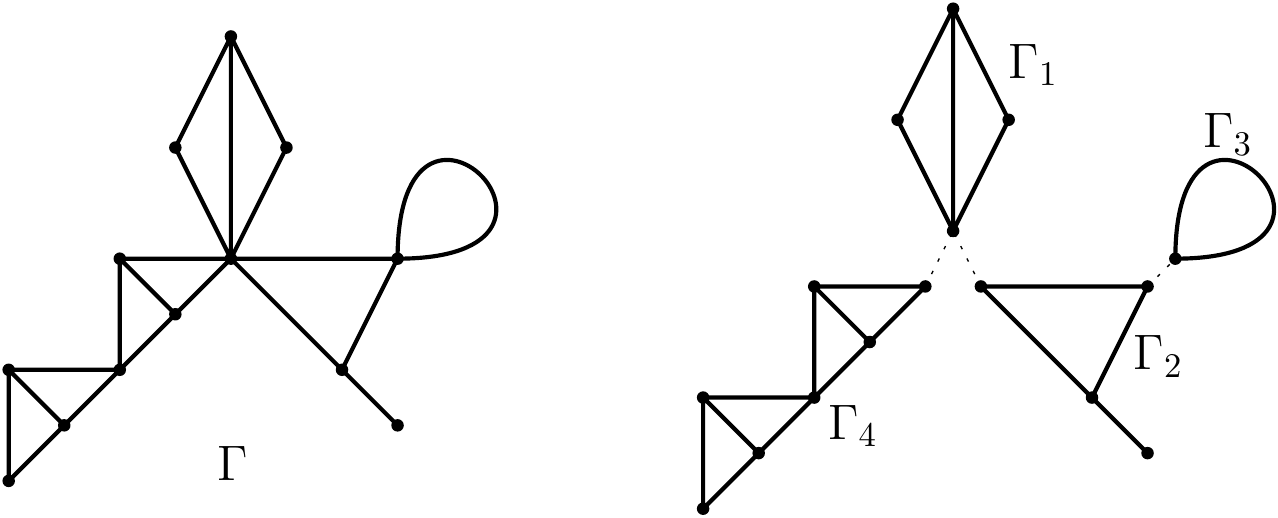}
  \caption{An example of a graph and its vertex separation.  Dashed
    lines indicate one possible choice for the introduction of
    zero-length edges that would make it an edge separation.  Note
    that the blocks $\Gamma_2$ and $\Gamma_4$ can be further decomposed.}
  \label{fig:vertex_separation}
\end{figure}

\begin{figure}
  \centering
  \includegraphics[scale=0.75]{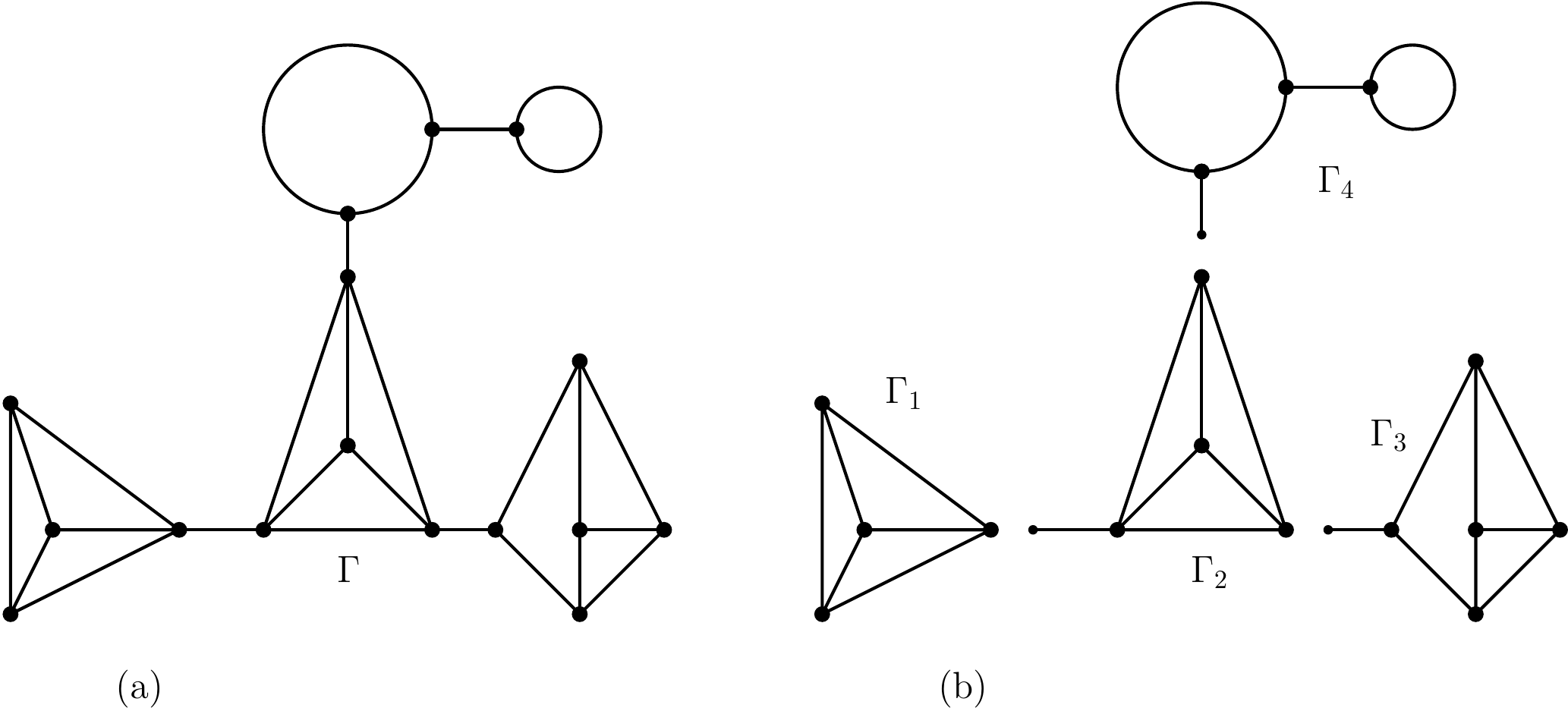}
  \caption{An example of a graph and its edge separation}
  \label{fig:edge_separation}
\end{figure}

\begin{remark}
  \label{rem:renumbering}
  While we defined a vertex separation as an ordered sequence of
  blocks, we have a freedom to choose an arbitrary block as the first
  one and reorder other blocks accordingly.  The same can be done with
  an edge separation but one may have to move bridges from one
  subgraph to another.
\end{remark}

One can ``convert'' a vertex-separated block decomposition into an
edge-separated one by introducing zero-length edges, see
Figure~\ref{fig:vertex_separation} for an example.  Note that in
this example the choice of zero-length edges is not unique.


Finally, the feature of separations which is most important for our
considerations is that they naturally partition the set of magnetic
fluxes, since the cycles of $\Gamma$ are precisely the cycles of its
blocks \cite[Lemma~3.1.1]{Diestel_graph_theory}.

To be more specific, in Section~\ref{sec:Reduction} we have defined
the standard representation of a magnetic operator with fluxes $\av$
by choosing a set $\Cut$ of $\beta$ edges whose removal does not
disconnect the graph and by placing the magnetic potential on these
edges only.  Each of the $\beta$ edges must belong to some $\Gamma_j$,
therefore each flux is naturally associated with one of $\Gamma_j$.

From now on we will assume that if we are given a separation (of
either kind) $[\Gamma_1,\ldots,\Gamma_n]$, the fluxes are ordered in
such a way that $\av = [\av_1,\ldots,\av_n]$, where the vector $\av_j$
contains all the fluxes corresponding to cycles in $\Gamma_j$.  In
particular, the dimension of $\av_j$ is equal to the first Betti
number $\beloc{j}$ of $\Gamma_j$.

\subsection{Block structure of the graph and local nodal surplus}
\label{sec:block_structure}

\begin{lemma}
  \label{lem:block_symmetry}
  Let $\Gamma$ be a magnetic standard graph with a vertex separation
  into two blocks and let $\av = [\av_1,\av_2]$ be the corresponding
  partition of the fluxes of $\Gamma$.  Let $F(\xv;\av)$ be the
  secular function of $\Gamma$.  Then for any $\xv$ and $\av_1$ we
  have
  \begin{equation}
    \label{eq:block_symmetry}
    F(\xv;\av_1,\av_2) = F(\xv; \av_1, -\av_2).
  \end{equation}
\end{lemma}

\begin{proof}
  We will be applying Theorem~\ref{thm:splitting_decomposition} with
  the connector set $\Cnct$ containing one edge which has zero length
  and no magnetic potential.  Both scattering matrices are
  one-dimensional, and therefore \eqref{eq:scat_mag_reversal} implies
  that $Z_2(\xv_2,-\av_2) = Z_2(\xv_2,\av_2)$.

  We just need to verify that $D_2$ is also even with respect to
  $\av$.  Similarly to the proof of
  Theorem~\ref{thm:scattering_properties}(\ref{item:scat_symmetry}), we employ a label
  switching matrix $J_2$ to write
  \begin{align*}
    D_2(\xv_2;\av_2) &= \det\left[J_2 \left(\Id_2 - e^{i\xv_2 +
          i\av_2}\Smat_2\right) J_2\right]
    = \det\left(\Id_2 - e^{i\xv_2 - i\av_2}\Smat_2^T\right) \\
    &= \det\left(\Id_2 - \Smat_2 e^{i\xv_2 - i\av_2}\right)^T
    = \det\left(\Id_2 - e^{i\xv_2 - i\av_2}\Smat_2\right)
    = D_2(\xv_2;-\av_2),
  \end{align*}
  where we used the identity $\det\left(\Id - AB\right) =
  \det\left(\Id - BA\right)$.

  Since all terms in \eqref{eq:splitting_decomposition} are even with
  respect to the change, $\av_2 \mapsto -\av_2$, the whole expression
  is even.
\end{proof}

\begin{theorem}
  \label{thm:block_Hessian}
  Let $\Gamma$ be a graph with a vertex separation
  $[\Gamma_1,\ldots,\Gamma_n]$, and let $\beloc{j}$ be the number of
  cycles in $\Gamma_j$.  Then the Hessian
  $H_{\av}(F)\left(\xv;\vec0\right)$ is block-diagonal with $j$-th
  block of size $\beloc{j}$.  In other words, if fluxes $\alpha_1$ and
  $\alpha_2$ belong to different vertex-separated blocks of the graph
  $\Gamma$ then
  \begin{equation}
    \label{eq:block_Hessian}
    \frac{\partial^2 F}{\partial \alpha_1 \partial
      \alpha_2}(\xv; \vec0) = 0 \qquad \mbox{for any }\xv.
  \end{equation}
\end{theorem}

\begin{proof}
  Let $\alpha_1$ and $\alpha_2$ belong to different blocks.  Then
  \begin{equation}
    \label{eq:proof_block_Hessian}
    \left. \frac{\partial^2 F}{\partial \alpha_1 \partial \alpha_2}
    \right|_{\av = 0}
    = \frac{\partial}{\partial \alpha_1} \left[
      \left. \frac{\partial F}{\partial \alpha_2}
      \right|_{\av_2=0}
    \right]_{\av_1=0}
    = 0,
  \end{equation}
  since $\frac{\partial F}{\partial \alpha_2} = 0$ at $\av_2=0$ by
  \eqref{eq:block_symmetry}.
\end{proof}

Simple examples of the block-diagonal structure of Hessian can be
found in Appendices~\ref{sec:figure8} and \ref{sec:dumbbell}.
The above theorem motivates the following definition.

\begin{definition}
  \label{def:local_surplus_fn}
  Let $\Gamma$ be a nontrivial standard graph with a vertex separation
  which induces the partition of fluxes
  $\av = [\av_1, \av_2, \ldots \av_n]$.  The \emph{local surplus
    functions} $\siloc{b}:\Sgen \rightarrow \{0,\ldots,\beloc{b}\}$
  are defined as follows:
  \begin{equation}
    \label{eq:local_surplus_fn_def}
    \siloc{b}\left(\xv\right)
    := \M\left[-\frac{H_{\av_b}(F)\left(\xv;\vec0\right)}{\vec{\nabla}F\cdot\lv}
    \right],
  \end{equation}
  where $\beloc{b}$ is the number of cycles in the block $\Gamma_b$
  or, equivalently, the number of entries in the vector $\av_b$.  We
  stress that the Hessian $H_{\av_b}(F)$ is taken with respect to the
  fluxes in $b$-th block only; it is a subblock of the full Hessian
  $H_{\av}(F)$, which is block-diagonal by Theorem~\ref{thm:block_Hessian}.
\end{definition}

\begin{remark}
  \label{rem:sum_of_surpluses}
  Observe that the summation of all local surplus functions gives the
  (total) surplus function,
  \begin{equation}
    \label{eq:local_surplus_additive}
    \sum_{b=1}^n \siloc{b} = \sigma.
  \end{equation}
  We also point out that the functions $\siloc{b}$ can be viewed as
  random variables on the probability space
  $\left(\Sgen,\bgm\right)$.  This will be convenient later when we
  talk about conditional probabilities and independence of $\siloc{b}$.
\end{remark}

\begin{remark}
  \label{rem:inversion_local}
  We have seen that the mapping $\inv: \xv \mapsto -\xv$ introduced in
  Lemma~\ref{lem:total_symmetry} changes the sign of the entire matrix
  appearing in the definition of the surplus functions.  Therefore,
  the conclusion applies to local surplus functions as well, namely
  \begin{equation}
    \label{eq:inversion_local}
    \siloc{b}(-\xv)
    = \beloc{b} - \siloc{b}(\xv),
  \end{equation}
  where $\beloc{b}$ is the number of cycles in the $b$-th block.
  Consequently, the distribution of the local surplus of the block $b$
  is symmetric around $\beloc{b}/2$.
\end{remark}

\subsection{Local surplus is determined by its block coordinates}
\label{sec:surplus_determined}
It is important to consider how much information is needed to
determine the value of a local surplus.  It turns out that for a block
with only one common vertex with the rest of the graph, the value of
the torus coordinates corresponding to the edges of the block are
enough --- together with the implicit information that $\xv$ lies on
the secular manifold.

\begin{figure}
  \centering
  \includegraphics[scale=0.85]{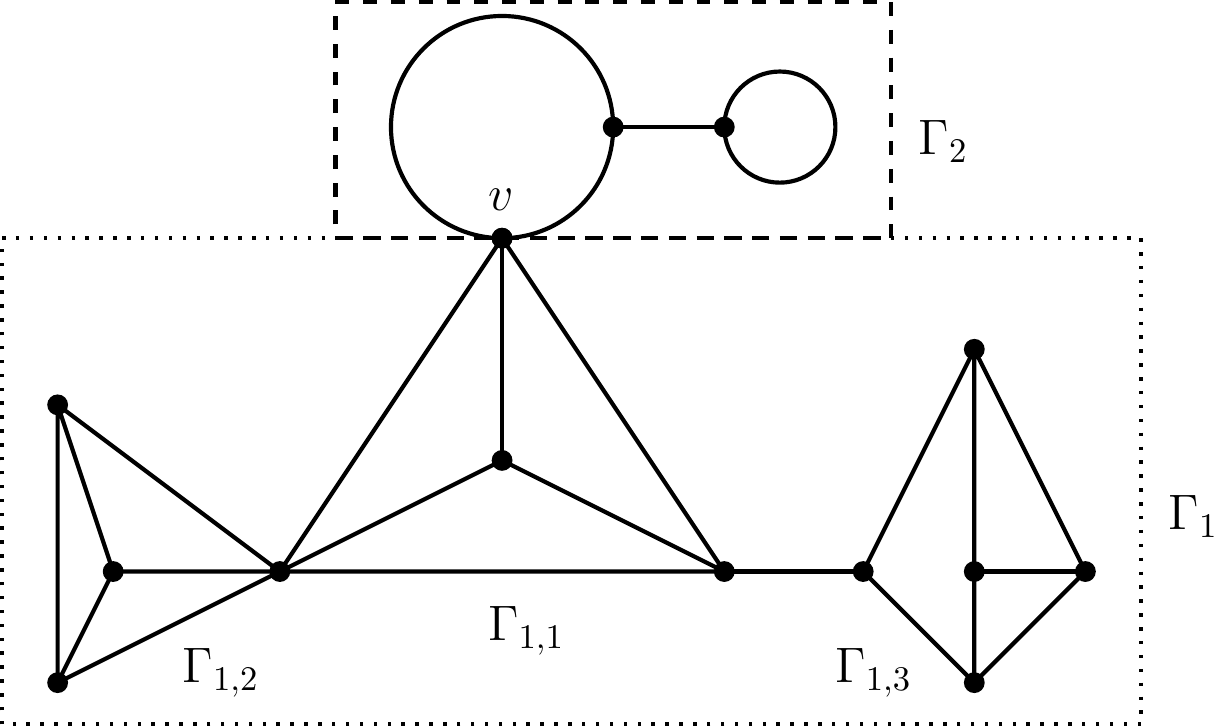}
  \caption{A graph with two vertex-separated blocks.  The block
    $\Gamma_1$ has a further separation into subblocks labeled
    $\Gamma_{1,j}$.}
  \label{fig:blocks_subblocks}
\end{figure}

\begin{lemma}
  \label{lem:surplus_determined}
  Let $\Gamma$ be a nontrivial standard graph with a vertex separation
  into two blocks, $[\Gamma_1,\Gamma_2]$.  Let $k_n$ be a generic
  eigenvalue.  Then the local surplus $\siloc 1$ at a point
  $\left(\xv_1,\xv_2\right)\in\Sgen$, is uniquely determined by the
  $\xv_1$ coordinates.

  More precisely if two points $\xv,\xv'\in\Sgen$ share
  the same $\xv_1$ coordinates, then
  \begin{equation}
    \label{eq:surplus_determined}
    \siloc1 \left(\xv\right) = \siloc1\left(\xv'\right).
  \end{equation}

  Furthermore, if $\Gamma_1$ has a vertex separation
  $[\Gamma_{1,1}, \Gamma_{1,2}, \ldots]$ (see
  Figure~\ref{fig:blocks_subblocks}) then for each subblock $j$
  \begin{equation}
    \label{eq:surplus_determined_subblock}
    \siloc{1,j} \left(\xv\right) = \siloc{1,j}\left(\xv'\right).
  \end{equation}
\end{lemma}

\begin{proof}
  We apply Theorem~\ref{thm:splitting_decomposition} with the connector set being just one edge of zero length from $v$ to $v$, where $v$ is the vertex common to $\Gamma_1$ and $\Gamma_2$. The secular
  equation is then written in the form
  \begin{align}
    \label{eq:sec_eq_bridge_rewrite}
    F(\xv; \av)
    = D_1(\xv_1;\av_1) D_2(\xv_2;\av_2)
    \Big(1 - Z_1(\xv_1;\av_1) Z_2(\xv_2; \av_2)\Big),
  \end{align}
  where $Z_j$ are $1\times1$ unitary matrices.

  Since $\xv\in\Sgen$, both determinants $D_j$ are
  non-zero, therefore, on the manifold $\Sgen$,
  \begin{equation}
    \label{eq:x1x2_relation}
    Z_2 = \frac{1}{Z_1} = \cc{Z_1}.
  \end{equation}

  We now use definition \eqref{eq:local_surplus_fn_def} in a modified
  form
  \begin{equation}
    \label{eq:local_surplus_fn_modified}
    \siloc 1\left(\xv\right)
    = \M\left[-\frac{H_{\av_1}(F)\left(\xv;\vec0\right)}
      {\vec{\nabla}_{\xv_1} F\cdot \lv_1}
      \right],
  \end{equation}
  where we used the fact that all entries of $\nabla F$ are of the
  same sign (up to a phase) and nonzero on $\Sgen$ (by (\ref{eq:normalS})
  and (\ref{eq:zero_Sigma})).
  When calculating this ratio of the Hessian of $F$ with respect to $\av_1$ and the
  gradient of $F$ with respect to $\xv_1$, the prefactor $D_2$ is
  canceled and $Z_2$ can be substituted with $\cc{Z_1}$ removing all
  dependence on $\xv_2$ variables.

  The second part of the statement follows immediately since the local
  surpluses of the subblocks of $\Gamma_1$ are determined as the
  Morse indices of the subblocks of the matrix in
  \eqref{eq:local_surplus_fn_modified} which we just determined to be
  independent of $\xv_2$ variables.
\end{proof}

\subsection{Local surplus of edge-separated blocks}
\label{sec:bridge_symmetry}

In this section we show that it is possible to ``localize'' the
homeomorphism of Lemma~\ref{lem:total_symmetry}: there is a mapping
that flips a given local surplus while keeping all other local
surpluses fixed.  We are able to establish this result only for
edge-separated blocks and we know from the example of
Appendix~\ref{sec:321pumpkin} that the result is not always true for
vertex separations.  To start, we need some simple facts about the
form of the eigenfunction on a bridge separating two blocks.

\begin{lemma}
  \label{lem:theta_determined}
  Let $\Gamma$ be a graph with an edge separation
  $[\Gamma_1,\Gamma_2]$ and let the bridge be denoted by $e_0$.
  For a $\xv \in \Sgen$ let the corresponding canonical eigenfunction
  (see Theorem~\ref{thm:Sreg_def_prop}(\ref{item:eigenvalue1})) on the
  edge $e_0$ be written in the form
  \begin{equation}
    \label{eq:eigen_on_bridge}
    f_{e_0} = C_{e_0} \cos\left(x - \theta_0(\xv)\right)
  \end{equation}
  on the bridge edge.  Then $\theta_0(\xv)$ is a smooth function on
  $\Sgen$ which is fully determined by the torus coordinates $\xv_1$
  corresponding to the edges of the block $\Gamma_1$.  In other
  words, if there is another point $\xv' \in \Sgen$ such that
  $\xv_1' = \xv_1$, then $\theta_0(\xv) = \theta_0(\xv')$.
\end{lemma}

\begin{proof}
  The above \emph{does not} mean that $\theta_0$ is independent of the
  other coordinates in $\xv$: the dependence is implicit via the
  relation $\xv \in \S$.  We can rephrase the result as saying that
  there exists a function $\Theta = \Theta(\xv_1)$ which is
  independent of $\xv_2$ and which on $\Sgen$ coincides with
  $\theta_0(\xv)$.

  Any solution of the subgraph $\Gamma_1$ and the bridge must belong
  to the set of scattering solutions on $\Gamma_1$ with the bridge
  extended to infinity.
  Comparing \eqref{eq:eigen_on_bridge} with equations \eqref{eq:lead_solution} and
  \eqref{eq:scattering_matrix_def}, we see that the function
  $\Theta(\xv_1)$ can be determined from
  $Z_1(\xv_1) = e^{2 i \Theta(\xv_1)}$, where $Z_1$ is the $1\times1$
  scattering matrix of the subgraph $\Gamma_1$, see
  Section~\ref{sec:scattering_def}.  Smoothness of $\Theta$ follows
  from Theorem~\ref{thm:scattering_properties}.
\end{proof}

\begin{lemma}
  \label{lem:bridge_symmetry}
  Let $\Gamma$ be a graph with an edge separation
  $[\Gamma_1,\Gamma_2]$ and the bridge denoted by $e_0$ (according to
  Definition~\ref{def:edge_separation}, $e_0$ belongs to the block
  $\Gamma_2$).  Let $\xv = (\xv_1, \tx_0, \xv_2)$ be the torus
  coordinates, where $\xv_1$ corresponds to the edges of $\Gamma_1$,
  $\tx_0$ corresponds to $e_0$ and $\xv_2$ corresponds to all other
  edges of $\Gamma_2$.  Consider the mapping
  \begin{equation}
    \label{eq:bridge_symmetry}
    R:\Sreg\rightarrow\Sreg,\qquad
    R\left(\xv_1, \tx_0, \xv_2 \right)
    = \left(\xv_1, -\tx_0 + 2\theta_0(\xv_1), -\xv_2\right),
  \end{equation}
  where $\theta_0(\xv_1)$ is a function whose existence is established
  in Lemma~\ref{lem:theta_determined}.  Then $R$ is a measure
  preserving homeomorphism of $\left(\Sgen,\bgm\right)$ to itself.

  Furthermore, the local surplus functions of the two subgraphs
  transform under $R$ according to
  \begin{equation}
    \label{eq:R_transforms_local_surplus}
    \siloc1 \circ R(\xv) = \siloc1(\xv) \qquad
    \siloc2 \circ R(\xv) = \beloc2 - \siloc2(\xv),
  \end{equation}
  where $\beloc2$ is the number of cycles in the subgraph $\Gamma_2$.
\end{lemma}

\begin{proof}
  The proof runs along the lines of the proof of
  Lemma~\ref{lem:total_symmetry} with some modifications.  The
  transformation $R$ is smooth (since $\theta_0$ is smooth) and
  invertible; its Jacobian matrix is triangular (because $\theta_0$ is
  determined by $\xv_1$) with $\pm1$ on the diagonal, therefore the
  Jacobian determinant is 1 in absolute value.  We are left to show
  the analogue of (\ref{eq:to_establish}) for $R$.

  Starting
  with an eigenfunction $f$ of eigenvalue 1 at the point
  $\left(\xv_1, \tx_0, \xv_2 \right)$ we will construct an
  eigenfunction $\tilde{f}$ at
  $\left(\xv_1, -\tx + 2\theta_0(\xv_1), -\xv_2\right)$.

  On the edges $e \in \Gamma_1$ we will set $\tilde{f}_e = f_e$.  On
  the edges $e\in\Gamma_2$ we let $\tilde{f}_e(x) = f_e(-x)$, where
  $f_e$ is understood to have been suitably extended (see equation
  (\ref{eq:function_form})).  Finally, on the bridge $e_0$ the
  function $f_{e_0}$ has the form
  \begin{equation*}
    f_{e_0}(x) = C_{e_0} \cos(x-\theta_0(\xv_1)),
  \end{equation*}
  where the variable $x$ is assumed to go from $x=0$ at the vertex
  common with the subgraph $\Gamma_1$ to $x=L_{e_0}$ at the vertex
  common with the subgraph $\Gamma_2$.  We again let $\tilde{f}_{e_0}(x) =
  f_{e_0}(x)$.

  As in the proof of Lemma~\ref{lem:total_symmetry}, the
  function values of $\tilde{f}_e$ for $e\in\Gamma_2$ remains the same
  while all derivatives change sign.  For $e\in\Gamma_1$, both
  function values and derivatives remain trivially the same.  We only
  need to check the function $\tilde{f}_{e_0}$ on the bridge.  At $x=0$ both the
  function value and the derivative is the same, fitting the rest of
  $\Gamma_1$.

  At the new edge end $x = -\tx_0 + 2\theta_0(\xv_1) =: \tilde{\tx}_0$ we have
  \begin{equation*}
    \tilde{f}_{e_0}\left(\tilde{\tx}_0\right)
    = C_{e_0} \cos(-\tx_0+\theta_0)
    = C_{e_0} \cos(\tx_0-\theta_0) = f_{e_0}(\tx_0)
  \end{equation*}
  while
  \begin{equation*}
    \tilde{f}_{e_0}'\left(\tilde{\tx}_0\right)
    = -C_{e_0} \sin(-\tx_0+\theta_0)
    = C_{e_0} \sin(\tx_0-\theta_0) = -f_{e_0}'(\tx_0).
  \end{equation*}
  This fits the values and the derivatives of the function on
  $\Gamma_2$ and therefore $\tilde{f}$ is an eigenfunction.

  As before, this construction preserves multiplicity and genericity
  of the eigenfunction and also preserves the density function of the
  Barra-Gaspard measure, $\left|\hat{n} \cdot \lv\right|$, by (\ref{eq:normalS}).

  By Remark~\ref{rem:inversion_local}, we have
  $\siloc2(-\xv) = \beloc2 - \siloc2(\xv)$.  Now we simply apply
  Lemma~\ref{lem:surplus_determined} to conclude
  \begin{equation*}
    \siloc1 \circ R(\xv) = \siloc1(\xv), \qquad
    \siloc2 \circ R(\xv) = \siloc1(-\xv) = \beloc2 - \siloc2(\xv),
  \end{equation*}
  since the surplus-determining coordinates agree in each case.
\end{proof}

\begin{remark}
  If the subgraph $\Gamma_2$ in Lemma~\ref{lem:bridge_symmetry}
  consists of several subblocks, the same conclusion applies to the local
  surpluses of the subblocks.  Namely, if a subblock $\Gamma_{2,j}$
  has $\beloc{2,j}$ cycles and $\siloc{2,j}$ is the corresponding
  surplus function, then with the transformation $R$ defined in
  \eqref{eq:bridge_symmetry} we have
  \begin{equation}
    \label{eq:R_transforms_subblocks}
    \siloc{2,j} \circ R(\xv)
    = \siloc{2,j}(-\xv) = \beloc{2,j} - \siloc{2,j}(\xv),
  \end{equation}
  where first equality follows from Lemma~\ref{lem:surplus_determined}
  and the second equality from Remark~\ref{rem:inversion_local}.
\end{remark}

\begin{corollary}
  \label{cor:one_block_flip}
  Let $\Gamma$ be a nontrivial standard graph with an edge separation
  $[\Gamma_1,\Gamma_2,\ldots,\Gamma_n]$.  There exist a measure
  preserving homeomorphism $\CR$ of $\left(\Sgen,\bgm\right)$ to
  itself, such that
  \begin{align}
    \label{eq:one_block_flip1}
    \siloc1 \circ \CR
    &= \beloc1 - \siloc1,\\
    \label{eq:one_block_flip2}
    \siloc{j}\circ\CR
    &= \siloc{j}, \quad j>1.
  \end{align}
\end{corollary}

\begin{proof}
  Assume initially that all blocks are connected directly to
  $\Gamma_1$ (this assumption is satisfied by the graph in
  Figure~\ref{fig:one_block_flip} and \emph{not} satisfied by the
  graph in Figure~\ref{fig:edge_separation}).
  Denote their bridges by $b_j$, $j=2,\ldots,n$.  For each $b_j$ we define a map
  $R_j:\Sgen\to\Sgen$ guaranteed by Lemma~\ref{lem:bridge_symmetry}
  such that the bridge of Lemma~\ref{lem:bridge_symmetry} is $b_j$,
  the second subgraph is $\Gamma_j$ and the first subgraph is the rest
  of the graph (including the block $\Gamma_1$); see
  Fig.~\ref{fig:one_block_flip} for an example.

  \begin{figure}
    \centering
    \includegraphics[scale=0.85]{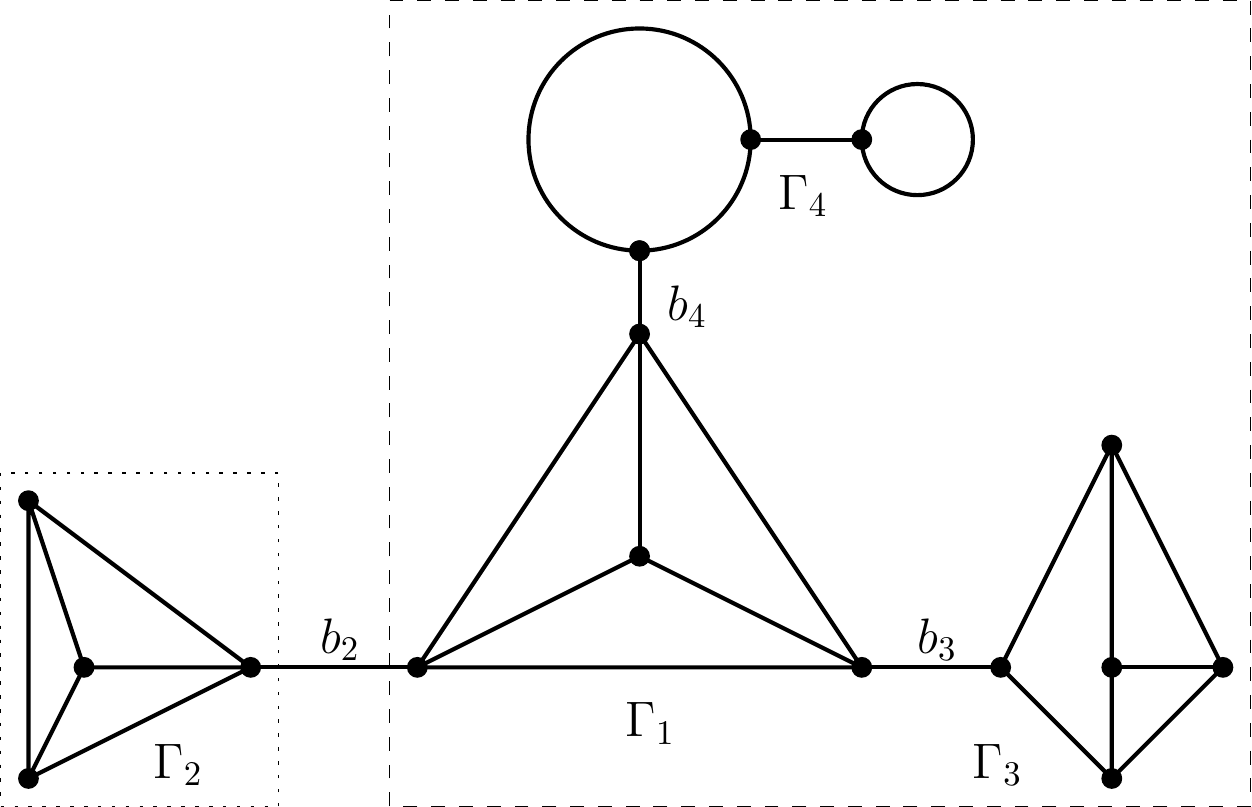}
    \caption{Construction of $R_2$ in the proof of
      Corollary~\ref{cor:one_block_flip}: subgraph within the dashed
      box is the first subgraph in the application of
      Lemma~\ref{lem:bridge_symmetry} and within the dotted box is the
      second.}
    \label{fig:one_block_flip}
  \end{figure}

  Each of the maps $\left\{R_j\right\}_{j=2}^n$ is a measure
  preserving transformation of
  $\left(\Sgen,\bgm\right)$ to itself.  Consider the
  map $\CR$, defined as follows
  \begin{equation*}
    \CR: = R_2\circ R_3\circ \ldots \circ R_n \circ \inv.
  \end{equation*}
  It is a measure preserving transformation of
  $\left(\Sgen,\bgm\right)$ to itself as a finite
  composition of such. Since each $R_j$ leaves $\siloc1$ invariant, and
  $\inv$ by Remark~\ref{rem:inversion_local} flips $\siloc1$ to
  $\beloc1-\siloc1$, we get
  \begin{equation*}
    \siloc1 \circ R_2 \circ R_3 \circ \ldots \circ R_n \circ \inv
    = \siloc1 \circ R_3 \circ \ldots \circ R_n \circ \inv = \ldots
    = \siloc1 \circ \inv = \beloc1 - \siloc1.
  \end{equation*}
  For the block number $j>1$, it is easy to see that $\CR \xv_j = \xv_j$
  (the sign gets flipped exactly twice, by $R_j$ and $\inv$) and
  equation~\eqref{eq:one_block_flip2} follows from
  Lemma~\ref{lem:surplus_determined}.  In fact this holds for all
  subblocks of $j$-th block, proving the Lemma in the general case
  (when some subblocks $\Gamma_j$ are connected to another subblock
  $\Gamma_{j'}$, $j>j'>1$ rather than directly to $\Gamma_1$).
\end{proof}

\subsection{Conditional surplus probabilities and the proof of
  Theorem~\ref{thm:binomial_distr}}
\label{sec:proof_main2}

We would now like to study the dependence among the local surpluses.  We
remind the reader that we can view the local nodal surplus functions
$\siloc1,\ldots,\siloc{n}$ as random variables on the probability space
$\left(\Sgen,\bgm\right)$ which automatically enables us to talk about
conditional probabilities.  In particular we will condition $\siloc{b}$
on all random variables $\siloc1, \ldots, \siloc{n}$ \emph{except}
$\siloc{b}$.  We will denote the latter set by
$\vec{\sigma}_{\hat{b}}$.

\begin{theorem}
  \label{thm:indep_symmetric}
  Let $\Gamma$ be a nontrivial standard graph with an edge separation
  $[\Gamma_1,\ldots,\Gamma_n]$ and let the block $\Gamma_b$ have
  $\beloc{b}$ cycles.  Then
  \begin{equation}
    \label{eq:indep_symmetric}
    \Prc{\siloc{b}=s}{\vec\sigma_{\hat{b}}}
    = \Prc{\siloc{b}= \beloc{b}-s}{\vec\sigma_{\hat{b}}}
  \end{equation}
  for all $s$, $0\leq s\leq \beloc{b}$.
  We say that the local surplus functions of $\Gamma$ are \emph{independently
  symmetric}.
\end{theorem}

\begin{proof}
  We renumber the blocks so that $b=1$ (see
  Remark~\ref{rem:renumbering}) and use
  Lemma~\ref{lem:bridge_symmetry} to construct the measure preserving
  transformation $\CR$ satisfying
  \eqref{eq:one_block_flip1}-\eqref{eq:one_block_flip2}.  Since $\CR$
  is measure preserving,
  \begin{equation*}
    \Prc{\siloc{b}=s}{\vec\sigma_{\hat{b}}}
    = \Prc{\siloc{b}\circ\CR = s}{\vec\sigma_{\hat{b}}\circ\CR}.
  \end{equation*}
  On the other hand, properties
  \eqref{eq:one_block_flip1}-\eqref{eq:one_block_flip2} imply that
  \begin{equation*}
    \Prc{\siloc{b}\circ\CR = s}{\vec\sigma_{\hat{b}}\circ\CR}
    = \Prc{\siloc{b} = \beloc{b} - s}{\vec\sigma_{\hat{b}}},
  \end{equation*}
  completing the proof.
\end{proof}

The previous theorem and the law of total probability immediately
yields the following corollary.

\begin{corollary}
  \label{cor:indep_sym_cycle}
  If $\beloc{b}=1$ then
  \begin{equation}
    \label{eq:indep_sym_b1}
    \Prc{\siloc{b}=0}{\vec\sigma_{\hat{b}}} =
    \Prc{\siloc{b}=1}{\vec\sigma_{\hat{b}}} = \frac12,
  \end{equation}
  That is, $\siloc{b}$ takes one of its two possible values with equal
  probabilities and independently of all other local surpluses.
\end{corollary}

\begin{proof}[Proof of Theorem~\ref{thm:binomial_distr}]
  In the setting of the Theorem, there is an edge separation where
  each block contains just one cycle of the graph.  It follows that
  $\left\{ \siloc{b}\right\}_{b=1}^{\beta}$ are independent Bernoulli
  random variables with $p=\frac{1}{2}$.  Therefore their sum $\sigma$
  has the binomial distribution with $p=\frac{1}{2}$ and $n=\beta$.
\end{proof}


\appendix

\section{Relative volume of $\Sgen$}

\label{sec:volume_Sgen}

The main theorems in this paper apply to generic eigenfunctions (i.e., those that do not vanish at vertices and correspond to a simple eigenvalue). It is therefore of interest to estimate the proportion of such eigenfunctions out of
the whole spectrum. The next proposition gives a precise geometric
expression for this ratio and shows that the majority of the eigenfunctions
are generic.

\begin{proposition} Let $\Gamma$ be a nontrivial standard graph
with rationally independent edge lengths $\lv$. Denote by $\genN$ the
set of indices $1\leq n\leq N$ such that the eigenvalue $k_{n}$
is generic. Then
\begin{equation}
d(\lv):=\lim_{N\to\infty}\frac{\left|\genN\right|}{N}=1-\frac{\mathcal{L}_{loops}}{2\mathcal{L}}\ge\frac{1}{2},\label{eq:proportion_generic}
\end{equation}
where $\mathcal{L}$ is the total length of the graph and $\mathcal{L}_{loops}$
is the total length of all loops (edges from a vertex to itself) in
the graph. \end{proposition}

\begin{proof}

Combining Lemma 3.1 in \cite{CdV_ahp15} together with Theorem \ref{thm:Sigma_generic}
in our paper we have that if the lengths $\lv$ are rationally independent
then
\begin{equation}
d(\lv)=\frac{\int_{\Sigma^{g}}\left|\hat{n}\cdot\lv\right|d\sigma}{\int_{\Sreg}\left|\hat{n}\cdot\lv\right|d\sigma}.\label{eq:alt_def_dl}
\end{equation}
Using (\ref{eq:alt_def_dl}) we may extend $d(\lv)$ to consider it
as a function which is defined for all $\lv\in\R_{+}^{E}$ and get
that it is continuous. It is therefore enough to prove (\ref{eq:proportion_generic})
for a residual set of lengths, which is what we do next.

Denote by $\mathcal{T}\subset\R_{+}^{E}$ the set of length vectors
$\lv$ for which the spectrum of the corresponding graph obeys both
of the following:
\begin{enumerate}
\item Every eigenvalue is simple.
\item Every eigenfunction which vanishes on one of the vertices is supported
on a single loop.
\end{enumerate}
Theorem 3.6 in \cite{BerLiu_jmaa17} ensures that $\mathcal{T}$ is
residual in $\R_{+}^{E}$.

Let us number the loop-edges of the graph by $\left\{ e_{1},...,e_{m}\right\} $
and define the counting functions
\begin{align*}
\mathcal{N}(K) & =\left\{ n:k_{n}<K\right\} \\
\mathcal{N}^{g}(K) & =\left\{ n:k_{n}<K\mbox{ and }k_{n}\mbox{ is generic}\right\} \\
\overline{\mathcal{N}}^{\left(j\right)}(K) & =\left\{ n:k_{n}<K,~k_{n}\mbox{ is simple and }f_{n}\mbox{ is supported only on }e_{j}\right\} ,\quad j=1,\ldots,m.
\end{align*}
Assuming $\lv\in\mathcal{T}$ , we get the following relation
\[
\mathcal{N}^{g}(K)=\mathcal{N}(K)-\sum_{j=1}^{m}\overline{\mathcal{N}}^{\left(j\right)}(K).
\]

Next, we use the Weyl asymptotics, $\mathcal{N}(K)\sim K\mathcal{L}/\pi$
to estimate the counting functions above. The eigenvalues producing
eigenfunctions supported only on the loop $e_{j}$ are precisely $k_{n}=2\pi n/l_{j}$,
$n\in\mathbb{N}$, so that their count is $\overline{\mathcal{N}}^{\left(j\right)}(K)\sim Kl_{j}/2\pi$.
We conclude that if $\lv\in\mathcal{T}$ then
\begin{equation}
\lim_{N\to\infty}\frac{\left|\genN\right|}{N}=\lim_{K\to\infty}\frac{\mathcal{N}^{g}(K)}{\mathcal{N}(K)}=\lim_{K\to\infty}\left(1-\sum_{j=1}^{m}\frac{\overline{\mathcal{N}}^{\left(j\right)}(K)}{\mathcal{N}(K)}\right)=1-\sum_{j=1}^{m}\frac{l_{j}/2}{\mathcal{L}}=1-\frac{\mathcal{L}_{loops}}{2\mathcal{L}},\label{eq:ratio_in_energy}
\end{equation}

which shows that (\ref{eq:proportion_generic}) holds for a residual set,
as required. \end{proof}

The proof above has a nice interpretation on the level of the secular manifold.
We present a decomposition of the secular manifold, which is schematically demonstrated
in Figure \ref{fig:Sigma_decomposition}. First, eigenfunctions
of simple eigenvalues correspond to $\Sreg$, the regular part of
the secular manifold (Theorem \ref{thm:Sreg_def_prop}). Out of those
eigenfunctions, the eigenfunctions which vanish at some vertex of
the graph correspond to $\Sigma_{0}$ (see (\ref{eq:zero_Sigma})).
We may further decompose $\Sigma_{0}$, by defining

\begin{align}
\Sigma_{F} & =\left\{ (x_{1},\ldots,x_{E})\in\Sreg~|~x_{e}=2\pi~\textrm{for some edge }e\textrm{, which is a loop}\right\} .\label{eq:Sigma_faces}
\end{align}

It is not hard to see that $\Sigma_{F}$ corresponds to simple eigenvalues
whose eigenfunctions are supported on a single loop of the graph.
Note that in (\ref{eq:BG_measure_def}) we could have defined the
Barra-Gaspard measure on the whole of $\Sreg$ (as is actually done
in \cite{CdV_ahp15}). Doing so we would get that $\Sigma_{0}\backslash\Sigma_{F}$
is of zero measure and that the total measure of $\Sigma_{F}$ is
$\frac{\mathcal{L}_{loops}}{2\mathcal{L}}$ (assuming that the total measure of $\Sreg$ is 1). In particular, it follows
that for graphs without loop-edges, $\Sigma^{g}=\Sreg$ up to measure
zero set.

\begin{figure}[h]
\includegraphics[scale=0.7]{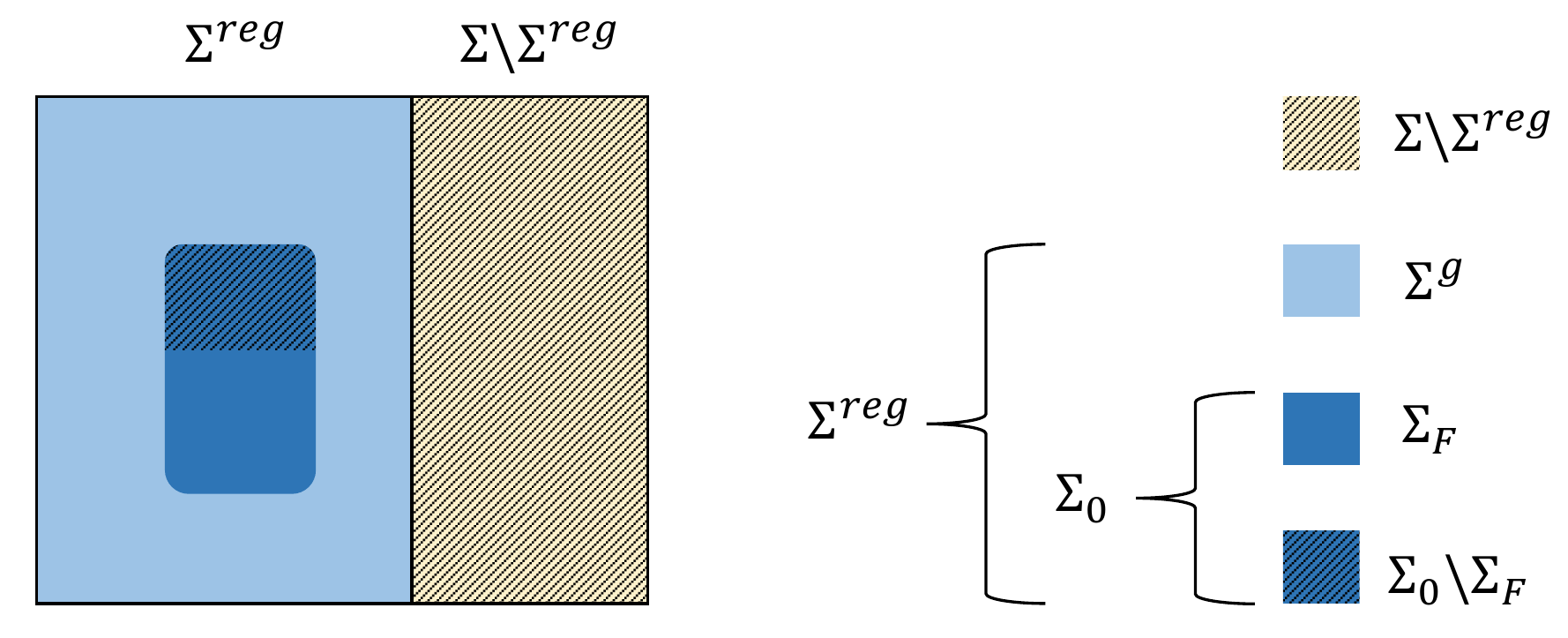} 
\caption{A schematic drawing which shows the relevant subsets of $\Sigma$.
The shaded regions $\Sigma\backslash\Sreg$ and $\Sigma_{0}\backslash\Sigma_{F}$
are subsets of $\Sigma$ which are of zero measure.}
\label{fig:Sigma_decomposition}
\end{figure}


\section{Contracting an edge in a graph}

\label{sec:zero_length_edge}

As pointed out by Band and L\'evy in \cite[Appendix A]{BanLev_prep16},
when the length of an edge tends to zero, the eigenvalues of the graph
converge to the eigenvalues of the graph with this edge contracted.
The edge to be contracted has Neumann conditions at its two end vertices
and upon contraction, those vertices are merged, and the Neumann conditions
are imposed on the newly formed vertex (see Figure \ref{fig:contracting_an_edge}).
Here, we consider what happens when the torus variable corresponding
to this edge is set to zero, which is needed for the proof of Theorem~\ref{thm:splitting_decomposition}
given in Appendix~\ref{sec:proof_splitting}.

\begin{lemma} \label{lem:contracting_an_edge} Let $\Gamma$ be a
magnetic standard graph. Let $e$ be an edge with no magnetic potential
on it and with distinct endpoints both endowed with Neumann conditions.
Let $\Gamma_{c}$ be the graph obtained from $\Gamma$ by contracting
the edge $e$ and imposing Neumann condition at the newly formed vertex.
Setting $\tx_{e}=0$ in the secular function of $\Gamma$ we obtain
\begin{equation}
F((0,\xv_{c});\av)=2\frac{d_{1}+d_{2}-2}{d_{1}d_{2}}F_{c}(\xv_{c};\av),\label{eq:contracted_secular_function_prefactor}
\end{equation}
where $F_{c}$ is the secular function of $\Gamma_{c}$ , and $d_{1},d_{2}$
are the degrees of the endpoints of $e$.\end{lemma}

\begin{remark} \label{rem:contracting_a_loop} If a graph $\Gamma$
has a loop and we set the corresponding variable to zero in the graph's
secular function, the secular function becomes identically zero. This
is why we explicitly assumed in Lemma~\ref{lem:contracting_an_edge}
that the edge to be contracted has distinct endpoints. \end{remark}

\begin{figure}
\centering \includegraphics{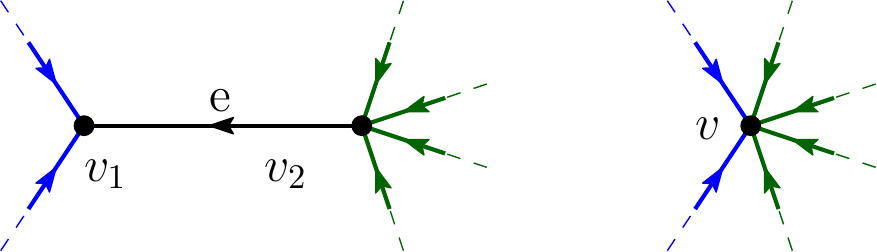} 
\caption{Contracting the edge $e$ in a graph. Here $d_{1}=3$ and $d_{2}=5$.}
\label{fig:contracting_an_edge}
\end{figure}

\begin{proof}{[}Proof of Lemma~\ref{lem:contracting_an_edge}{]}
Assume we are contracting an edge connecting vertices $v_{1}$ and
$v_{2}$ of degrees $d_{1}$ and $d_{2}$ correspondingly, see Figure~\ref{fig:contracting_an_edge}.
Let $e$ refer to the directed label of this edge going towards $v_{1}$
and $\hat{e}$ denote its reversal. Upon contraction of the edge $e$
the new joined vertex $v=v_{1}=v_{2}$ will have the degree $d_{1}+d_{2}-2$.
The new graph will be denoted $\Gamma_{c}$ and its torus coordinates
are $\xv_{c}$.

Further assume that $e$ comes first in the numbering of edges (\ref{eq:vector_amplitudes})
used in the set up of the secular equation (\ref{eq:sec_condition}).
As a result, $\xv=(\tx_{1},\xv_{c})$. Once we set $\tx_{1}=0$, the
matrix used in the definition of the secular function $F(\xv;\av)$
takes the form
\begin{equation}
\Id-e^{i(\xv+\av)}\Smat=\left(\begin{array}{cc|c}
1 & 1-\frac{2}{d_{2}} & -\frac{2}{d_{2}}\vec{r}_{2}\\
1-\frac{2}{d_{1}} & 1 & -\frac{2}{d_{1}}\vec{r}_{1}\\
\hline -\frac{2}{d_{1}}e^{i\xi}\vec{c}_{1} & -\frac{2}{d_{2}}e^{i\xi}\vec{c}_{2} & \Id-e^{i\xi}\Smat_{r}
\end{array}\right),\label{eq:matrix_with_L0}
\end{equation}
where $\xi=(\tx_{2}+\alpha_{2},\tx_{2}-\alpha_{2},\tx_{3}+\alpha_{3},\ldots)$
is the vector of torus coordinates and fluxes of all edges except
the edge $e$ and $\Smat_{r}$ is the matrix $\Smat$ without its
first two rows and columns. We have the relation $\vec{r}_{j}=(Jc_{j})^{T}$,
where $J$ is a matrix which switches around the directed labels of
the edges (see equation (\ref{eq:scattering problem}) and preceding
discussion). Furthermore, the vectors $\vec{c}_{j}$ and $\vec{r}_{j}=(Jc_{j})^{T}$
contain zeros and ones only; for example, $\vec{r}_{2}$ has ones
only in the entries corresponding to the edge labels coming into vertex
$v_{2}$ while $\vec{c}_{1}$ has ones corresponding to edge labels
coming out of $v_{1}$.

We would like to evaluate the determinant of (\ref{eq:matrix_with_L0})
using Schur's determinant identity in the form
\begin{equation}
\det\begin{pmatrix}A & B\\
C & D
\end{pmatrix}=\det(A)\det\left(D-CA^{-1}B\right).\label{eq:Shurs_identity}
\end{equation}
We get
\begin{equation}
F((0,\xv_{c});\av)=2\frac{d_{1}+d_{2}-2}{d_{1}d_{2}}\det\left(\Id-e^{i\xi}\Smat_{c}\right),\label{eq:contracted_secular_function_prefactor_preliminary}
\end{equation}
where the prefactor is the determinant of the top left corner and
\begin{align*}
\Smat_{c} & =\Smat_{r}+\begin{pmatrix}\vec{c}_{1} & \vec{c}_{2}\end{pmatrix}\begin{pmatrix}\frac{2}{d_{1}} & 0\\
0 & \frac{2}{d_{2}}
\end{pmatrix}\begin{pmatrix}1 & 1-\frac{2}{d_{2}}\\
1-\frac{2}{d_{1}} & 1
\end{pmatrix}^{-1}\begin{pmatrix}\frac{2}{d_{2}} & 0\\
0 & \frac{2}{d_{1}}
\end{pmatrix}\begin{pmatrix}\vec{r}_{2}\\
\vec{r}_{1}
\end{pmatrix}\\
 & =\Smat_{r}+\begin{pmatrix}\vec{c}_{1} & \vec{c}_{2}\end{pmatrix}\begin{pmatrix}\frac{2}{d_{1}+d_{2}-2} & \frac{2}{d_{1}+d_{2}-2}-\frac{2}{d_{1}}\\
\frac{2}{d_{1}+d_{2}-2}-\frac{2}{d_{2}} & \frac{2}{d_{1}+d_{2}-2}
\end{pmatrix}\begin{pmatrix}\vec{r}_{2}\\
\vec{r}_{1}
\end{pmatrix}.
\end{align*}
The effect of adding the term is best understood on some examples.
If $e_{in}$ is one of the edges going to $v_{2}$ and $e_{out}$
is one of the edges coming out of $v_{1}$, then
\[
(\Smat_{c})_{e_{out},e_{in}}=0+\begin{pmatrix}1 & 0\end{pmatrix}\begin{pmatrix}\frac{2}{d_{1}+d_{2}-2} & \frac{2}{d_{1}+d_{2}-2}-\frac{2}{d_{1}}\\
\frac{2}{d_{1}+d_{2}-2}-\frac{2}{d_{2}} & \frac{2}{d_{1}+d_{2}-2}
\end{pmatrix}\begin{pmatrix}1\\
0
\end{pmatrix}=\frac{2}{d_{1}+d_{2}-2}.
\]
If $e_{in}$ is one of the edges going to $v_{1}$ and $e_{out}$
is its reversal,
\[
(\Smat_{c})_{e_{out},e_{in}}=\frac{2}{d_{1}}-1+\begin{pmatrix}1 & 0\end{pmatrix}\begin{pmatrix}\frac{2}{d_{1}+d_{2}-2} & \frac{2}{d_{1}+d_{2}-2}-\frac{2}{d_{1}}\\
\frac{2}{d_{1}+d_{2}-2}-\frac{2}{d_{2}} & \frac{2}{d_{1}+d_{2}-2}
\end{pmatrix}\begin{pmatrix}0\\
1
\end{pmatrix}=\frac{2}{d_{1}+d_{2}-2}-1.
\]
In both cases the answer is the correct scattering amplitude for the
vertex $v$ of degree $d_{1}+d_{2}-2$ which resulted from the contraction
of $e$ (see, for example, Figure~\ref{fig:contracting_an_edge},
right). The remaining cases are checked analogously and we find that
$\Smat_{c}$ is the bond scattering matrix of the graph $\Gamma_{c}$,
so that
\[
F_{c}(\xv_{c};\av)=\det\left(\Id-e^{i\xi}\Smat_{c}\right).
\]
This, together with (\ref{eq:contracted_secular_function_prefactor_preliminary})
gives (\ref{eq:contracted_secular_function_prefactor}).\end{proof}

We note that Lemma \ref{lem:contracting_an_edge} indeed implies the
claim in \cite[Appendix A]{BanLev_prep16}, that graph eigenvalues
are continuous with respect to edge length when the length goes to
zero. Next, we provide another proof for this statement. The proof
is insightful as it is done via the eigenfunctions. Yet, it does not
reproduce the exact prefactor given in (\ref{eq:contracted_secular_function_prefactor}).
More precisely, we now prove that
\begin{equation}
F\left(0,\xv_{c};\av\right)=0\iff F_{c}\left(\xv_{c};\av\right)=0,\label{eq:zeros_of_secular_functions_coincide}
\end{equation}
 with notations similar to those in Lemma \ref{lem:contracting_an_edge}.

First, we generalize the notion of canonical eigenfunction as given
in Theorem \ref{thm:Sreg_def_prop}(\ref{item:eigenvalue1}). We do
so by taking a canonical eigenfunction to be any eigenfunction belonging
to the eigenvalue $k=1$ for a graph with edge lengths given by $\lv=\xv\in\Sigma$
(not just for $\xv\in\Sreg$ as in Theorem \ref{thm:Sreg_def_prop}(\ref{item:eigenvalue1})).
Now, the proof is based on showing a one to one correspondence between
canonical eigenfunctions of $\Gamma$ with edge lengths $\lv=\left(2\pi,\xv_{c}\right)$
and canonical eigenfunctions of $\Gamma_{c}$ with edge lengths $\lv_{c}=\xv_{c}$.
Let $f$ be a canonical eigenfunction of $\Gamma$ for edge lengths
$\lv=\left(2\pi,\xv_{c}\right)$. Denote its restriction $\tilde{f}=f|_{\Gamma\setminus e_{0}}$
and consider $\tilde{f}$ to be a function on $\Gamma_{c}$ with $\lv_{c}=\xv_{c}$
under the identification of $v_{1}\sim v_{2}$. We will show that
$\tilde{f}$ is actually a canonical eigenfunction of $\Gamma_{c}$.
Note that we may write $f$ as
\begin{equation}
f_{e}\left(t\right)=\begin{cases}
\tilde{f}_{e}\left(t\right) & e\ne e_{0}\\
B\cos\left(t+\theta\right) & e=e_{0}
\end{cases},\label{eq:canonical_eigenfunction_on_Gamma}
\end{equation}
 for some values of $B,\theta$. First of all, since $e_{0}$ is not
a loop, $\tilde{f}$ cannot be identically zero. Otherwise the function
$f$ would have zero value and zero derivative at $v_{1}$ and therefore
would be identically zero on the edge $e_{0}$ as well. Denote the
set of edge labels coming into $v_{i}$ and not including the labels
$e_{0}$, $\hat{e}_{0}$ by $\E_{i}$ (blue and green edges in Figure~\ref{fig:contracting_an_edge}).
Let $e_{0}$ be oriented from $v_{2}$ to $v_{1}$ (see Figure~\ref{fig:contracting_an_edge}).
The Neumann conditions on the vertices imply that
\begin{align*}
\forall e\in\E_{1}\quad\quad & f_{e}\left(l_{e}\right)=f_{e_{0}}\left(2\pi\right)=B\cos\left(\theta\right)\\
\forall e\in\E_{2}\quad\quad & f_{e}\left(l_{e}\right)=f_{e_{0}}\left(0\right)=B\cos\left(\theta\right)\\
\sum_{e\in\E_{1}}\left(i\Diff{}{x}+\frac{\alpha_{e}}{l_{e}}\right)f_{e}\left(l_{e}\right) & =-i\Diff{}{x}f_{e_{0}}\left(2\pi\right)=iB\sin\left(\theta\right)\\
\sum_{e\in\E_{2}}\left(i\Diff{}{x}+\frac{\alpha_{e}}{l_{e}}\right)f_{e}\left(l_{e}\right) & =i\Diff{}{x}f_{e_{0}}\left(0\right)=-iB\sin\left(\theta\right),
\end{align*}
where for convenience we have chosen the magnetic potential to be
constant along each edge, $A_{e}\equiv\frac{\alpha_{e}}{l_{e}}$ (see
(\ref{eq:current_cons_mag})). It follows that
\begin{align*}
\forall e\in\left(\E_{1}\cup\E_{2}\right)\quad\quad & f_{e}\left(l_{e}\right)=B\cos\left(\theta\right)\\
\sum_{e\in\left(\E_{1}\cup\E_{2}\right)}\left(i\Diff{}{x}+\frac{\alpha_{e}}{l_{e}}\right)f_{e}\left(l_{e}\right) & =0.
\end{align*}
We get that if $f$ satisfies the Neumann boundary conditions at $v_{1}$
and $v_{2}$ then $\tilde{f}$ satisfies the Neumann boundary
conditions at the merged vertex $v$. Obviously, $\tilde{f}$ satisfies
the same vertex conditions as $f$ at all other vertices and it obeys
$-\tilde{f}''=\tilde{f}$. Therefore $\tilde{f}$ is a canonical eigenfunction
of $\Gamma_{c}$ with $\lv=\xv$.

In the other direction, assume that $\tilde{f}$ is a canonical eigenfunction
of $\Gamma_{c}$ with edge lengths $\lv_{c}=\xv_{c}$. It can be extended
to a canonical eigenfunction $f$ on $\Gamma$ with edge lengths $\lv=\left(2\pi,\xv_{c}\right)$,
by setting $f$ to be as in (\ref{eq:zeros_of_secular_functions_coincide}),
where $B,\theta$ are chosen to satisfy
\begin{align*}
B\cos\left(\theta\right) & =\tilde{f}\left(v\right)\\
iB\sin\left(\theta\right) & =\sum_{e\in\E_{1}}\left(i\Diff{}{x}+\frac{\alpha_{e}}{l_{e}}\right)\tilde{f}|_{e}\left(l_{e}\right).
\end{align*}
This proves (\ref{eq:zeros_of_secular_functions_coincide}). Furthermore,
from the above also follows that the multiplicities of corresponding
zeros of $F\left(0,\xv_{c};\av\right)$ and $F_{c}\left(\xv_{c};\av\right)=0$
are equal.


\section{Proof of Theorem~\ref{thm:splitting_decomposition}}

\label{sec:proof_splitting}

\begin{proof}[Proof of Theorem~\ref{thm:splitting_decomposition}] Let us have an in-depth look at the bond scattering
matrix $\Smat$ of the graph $\Gamma$. We order the directed edge
labels as follows: edge labels of $\Gamma_{1}$, edge labels of $\Cnct$
in the direction from $\Gamma_{1}$ to $\Gamma_{2}$, edge labels
of $\Cnct$ in the opposite direction and then edge labels of $\Gamma_{2}$.
With this order and edge groupings, the matrix $\Smat$ has the following
block structure
\begin{equation}
\Smat=\begin{pmatrix}\Smat_{1} & 0 & t_{1} & 0\\
t_{1}' & 0 & r_{1} & 0\\
0 & r_{2} & 0 & t_{2}'\\
0 & t_{2} & 0 & \Smat_{2}
\end{pmatrix},\label{eq:structure_Smat_bridge}
\end{equation}
where, for example, the matrix $t_{1}$ corresponds to scattering
of waves from $\Cnct$ into the subgraph $\Gamma_{1}$ and $r_{1}$
represents reflection of the waves from $\Cnct$, off $\Gamma_{1}$
and back into $\Cnct$. We have the relations $t_{i}'=(J_{i}t_{i})^{T}$,
where the permutation matrix $J_{i}$ switches the orientation of
the edge labels in the subgraph $\Gamma_{i}$ (see equation (\ref{eq:scattering problem})
and preceding discussion).

We now multiply the matrix $\Smat$ by the diagonal matrix $e^{i\xv+i\av}$
which, in the block form similar to \eqref{eq:structure_Smat_bridge},
is given by
\begin{equation}
e^{i\xv+i\av}=\begin{pmatrix}e^{i\xv_{1}+i\av_{1}} & 0 & 0 & 0\\
0 & e^{i\tx_{0}+i\av_{0}} & 0 & 0\\
0 & 0 & e^{i\tx_{0}-i\av_{0}} & 0\\
0 & 0 & 0 & e^{i\xv_{2}+i\av_{2}}
\end{pmatrix}.\label{eq:matrix_e}
\end{equation}
We would like to apply Schur's determinant identity
\begin{equation}
\det\begin{pmatrix}A & B\\
C & D
\end{pmatrix}=\det(D)\det\left(A-BD^{-1}C\right),
\label{eq:Shurs_identity_again}
\end{equation}
to the determinant of the matrix $\Id-e^{i\xv+i\av}\Smat$ written
as
\begin{equation}
\Id-e^{i\xv+i\av}\Smat=\left(\begin{array}{ccc|c}
\Id_{1}-e^{i\xv_{1}+i\av_{1}}\Smat_{1} & 0 & -e^{i\xv+i\av_{1}}t_{1} & 0\\
-e^{i\xv_{0}+i\av_{0}}t_{1}' & \Id_{0} & -e^{i\xv_{0}+i\av_{0}}r_{1} & 0\\
0 & -e^{i\xv_{0}-i\av_{0}}r_{2} & \Id_{0} & -e^{i\xv_{0}-i\av_{0}}t_{2}'\\
\hline 0 & -e^{i\xv_{2}+i\av_{2}}t_{2} & 0 & \Id_{2}-e^{i\xv_{2}+i\av_{2}}\Smat_{2}
\end{array}\right).\label{eq:funny_matrix}
\end{equation}
The block $D$ is going to be $\Id_{2}-e^{i\xv_{2}+i\av_{2}}\Smat_{2}$
and, as a first step, we would like to determine when it is invertible.
Suppose, for a choice of $\xv_{2}$ and $\av_{2}$, the vector $\vec{v}_{2}$
is an eigenvector of $e^{i\xv_{2}+i\av_{2}}\Smat_{2}$ with eigenvalue
$1$. Then the vector
\begin{equation}
\vec{v}=(0,0,0,\vec{v}_{2}^{T})^{T}\label{eq:deg_eigenvector}
\end{equation}
is the eigenvector of $e^{i\xv+i\av}\Smat$ with eigenvalue $1$.
Indeed, we know that the last entry $e^{i\xv+i\av}\Smat\vec{v}$ is
going to be $\vec{v}_{2}$ and if any other entry is non-zero, it
would mean that an application of $e^{i\xv+i\av}\Smat$ increases
the norm of $\vec{v}$ which is impossible for a unitary matrix.

We conclude that $\Id_{2}-e^{i\xv_{2}+i\av_{2}}\Smat_{2}$ is non-invertible
only when the graph has an eigenfunction vanishing on the connector
set $\Cnct$ and the entire subgraph $\Gamma_{1}$. When this is the
case we have that (\ref{eq:splitting_decomposition}) holds with $D_{2}=0$.
To add on that, when this happens for $\av=\vec{0}$, the eigenfunction
mentioned above implies $\xv\notin\Sgen$, which proves the claim
at the end of the theorem.

Assuming the matrix is invertible, we have for $BD^{-1}C$
\[
\begin{pmatrix}0\\
0\\
-e^{i\tx_{0}-i\av_{0}}t_{2}'
\end{pmatrix}\left(\Id_{2}-e^{i\xv_{2}+i\av_{2}}\Smat_{2}\right)^{-1}\begin{pmatrix}0 & -e^{i\xv_{2}+i\av_{2}}t_{2} & 0\end{pmatrix}=\begin{pmatrix}0 & 0 & 0\\
0 & 0 & 0\\
0 & e^{i\tx_{0}-i\av_{0}}(Z_{2}-r_{2}) & 0
\end{pmatrix},
\]
where $Z_{2}$ is given
\begin{equation}
Z_{2}(\xv_{2};\av_{2})=r_{2}+t_{2}'\left(\Id_{2}-e^{i\xv_{2}+i\av_{2}}\Smat_{2}\right)^{-1}e^{i\xv_{2}+i\av_{2}}t_{2},\label{eq:Z_2_expression}
\end{equation}
which coincides with the definition of the scattering matrix of the
subgraph $\Gamma_{2}$, see \eqref{eq:scattering_torus}.

Subtracting this from the block $A$ we get
\begin{multline*}
\det\left(\Id-e^{i\xv+i\av}\Smat\right)=\det\left(\Id_{2}-e^{i\xv_{2}+i\av_{2}}\Smat_{2}\right)\det\begin{pmatrix}\Id_{1}-e^{i\xv_{1}+i\av_{1}}\Smat_{1} & 0 & -e^{i\xv+i\av_{1}}t_{1}\\
-e^{i\tx_{0}+i\av_{0}}t_{1}' & \Id_{0} & -e^{i\xv_{0}+i\av_{0}}r_{1}\\
0 & -e^{i\xv_{0}-i\av_{0}}Z_{2} & \Id_{0}
\end{pmatrix}.
\end{multline*}

Applying Schur's determinant identity again now with $\Id_{1}-e^{i\xv_{1}+i\av_{1}}\Smat_{1}$
acting as a factor to bring outside, we get
\begin{multline*}
\det\left(\begin{array}{c|cc}
\Id_{1}-e^{i\xv_{1}+i\av_{1}}\Smat_{1} & 0 & -e^{i\xv+i\av_{1}}t_{1}\\
\hline -e^{i\tx_{0}+i\av_{0}}t_{1}' & \Id_{0} & -e^{i\xv_{0}+i\av_{0}}r_{1}\\
0 & -e^{i\xv_{0}-i\av_{0}}Z_{2} & \Id_{0}
\end{array}\right)\\
=\det\left(\Id_{1}-e^{i\xv_{1}+i\av_{1}}\Smat_{1}\right)\det\begin{pmatrix}\Id_{0} & -e^{i\xv_{0}+i\av_{0}}Z_{1}\\
-e^{i\xv_{0}-i\av_{0}}Z_{2} & \Id_{0}
\end{pmatrix},
\end{multline*}
where in the above we used an expression of $Z_{1}$ similar to (\ref{eq:Z_2_expression})
and also assumed the invertibility of $\Id_{1}-e^{i\xv_{1}+i\av_{1}}\Smat_{1}$.
If it is not invertible, we get just as before that $D_{1}=0$ and
(\ref{eq:splitting_decomposition}) still holds. Evaluating the last
determinant (using Schur's identity once more), we get
\[
\det\begin{pmatrix}\Id_{0} & -e^{i\xv_{0}+i\av_{0}}Z_{1}\\
-e^{i\xv_{0}-i\av_{0}}Z_{2} & \Id_{0}
\end{pmatrix}=\det(\Id_{0})\det\left(\Id_{0}-e^{i\xv_{0}+i\av_{0}}Z_{1}e^{i\xv_{0}-i\av_{0}}Z_{2}\right).
\]
If all entries of $\xv_{0}$ are different than zero, collecting all
the factors above together gives (\ref{eq:splitting_decomposition})
with the prefactor $c=1$ at its right hand side. Otherwise, for any
vanishing entry of $\xv_{0}$, we apply Lemma \ref{lem:contracting_an_edge}
to conclude that (\ref{eq:splitting_decomposition}) still holds.
This time, the prefactor $c$ in (\ref{eq:splitting_decomposition})
equals the product of prefactors at the right hand side of (\ref{eq:contracted_secular_function_prefactor}),
applied for each of the vanishing entries in $\xv_{0}$. \end{proof}


\section{Examples of nodal surplus distribution}

\label{sec:examples}

In this appendix we will calculate, analytically or numerically, the
nodal surplus distribution of the graphs shown in the Figure~\ref{fig:appendix_examples}.
We say that a graph is a $[p_{1},p_{2},\ldots,p_{n}]$ chain if the
graph consists of a sequence of $n+1$ vertices with $p_{k}$ edges
connecting vertices $v_{k}$ and $v_{k+1}$. Note that the graphs
we call ``figure of 8'' (Figure~\ref{fig:appendix_examples}(a))
and ``dumbbell'' (Figure~\ref{fig:appendix_examples}(b)) can be
considered as a $[2,2]$ chain and a $[2,1,2]$ chain correspondingly.
This terminology for ``chains'' comes from the notion of ``mandarin
chain'' or ``pumpkin chain'' graphs which appeared in \cite{K2M2_ahp16,BanLev_prep16}.

\begin{figure}
\includegraphics{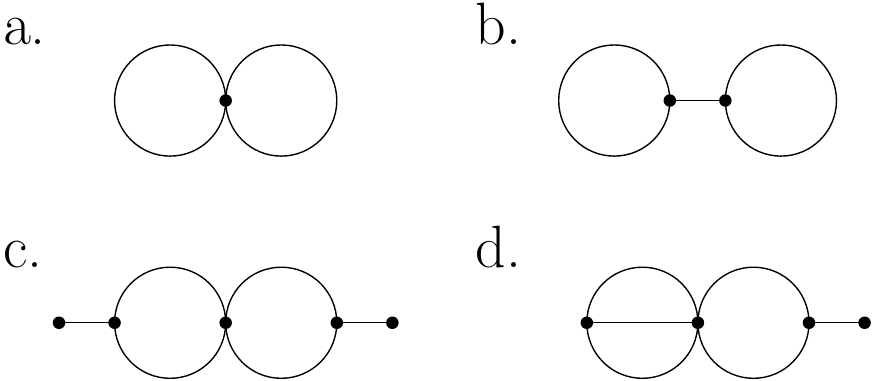} 
\caption{Graphs considered in Section \ref{sec:examples}: (a) ``figure of
8'', (b) ``dumbbell'', (c) ``$[1,2,2,1]$ chain'', (d) ``$[3,2,1]$
chain''.}
\label{fig:appendix_examples}
\end{figure}

The results are summarized in Table~\ref{tab:summary_of_numerics}.
Example (b) satisfies the assumptions of Theorem~\ref{thm:binomial_distr}
and therefore has a binomial distribution, which we confirm both analytically
and numerically. Examples (a), (d) show that these assumptions are
essential as there are graphs with non-binomial distribution. Example
(c) shows that not all graphs with binomial distribution can be characterized
as edge-separated, so the interesting question of characterizing all
the graphs with binomial distributions remains open.

\begin{table}
  \centering
  \begin{tabular}{|c||c|c|}
    \hline
    & Edge-separated
    & \parbox[c][2em][c]{0.3\textwidth}{\centering Not edge-separated} \\
    \hline \hline
    Binomial distribution
    & \parbox[c][2em][c]{0.3\textwidth}{\centering (b) ``dumbbell''}
    & \parbox[c][2em][c]{0.3\textwidth}{\centering (c) ``$[1,2,2,1]$ chain''}\\
      \hline
      Non-binomial distribution
    &
    & \parbox[c][5em][c]{0.3\textwidth}{\centering (a) ``figure of
      8'',\\
    (d) ``$[3,2,1]$ chain''}\\
    \hline
  \end{tabular}
  \caption{A summary of the results of Section~\ref{sec:examples}.}
  \label{tab:summary_of_numerics}
\end{table}

\subsection{A figure of 8 graph}
\label{sec:figure8}

Consider the graph shown in Figure~\ref{fig:appendix_examples}(a)
(A figure of 8 graph). The torus which describes this graph is $\T^{2}=\R^{2}/(2\pi\Z)^{2}$.
Using the coordinates $\left(x_{1},x_{2}\right)\in\T^{2}$, the real
secular function~(\ref{eq:Real_sec_func}) can be calculated to be
\begin{equation}
F_{R}\left(x_{1},x_{2};\alpha_{1},\alpha_{2}\right)=2\left(\cos\alpha_{2}\sin x_{1}+\cos\alpha_{1}\sin x_{2}-\sin\left(x_{1}+x_{2}\right)\right).\label{eq:FR_figure8}
\end{equation}
For $\av=0$ we get a nice form of
\begin{align*}
F_{R}\left(x_{1},x_{2};0,0\right) & =8\sin\left(\frac{x_{1}}{2}\right)\sin\left(\frac{x_{2}}{2}\right)\sin\left(\frac{x_{1}+x_{2}}{2}\right).
\end{align*}
It is not hard to show that (see also Figure \ref{fig:figure_8_torus}),
\begin{align*}
\Sreg & =\left\{ \left(x_{1},x_{2}\right)\in\T^{2}~:~\,x_{1}=0~\textrm{or}~x_{2}=0~\textrm{or}~x_{1}+x_{2}\equiv0\thinspace\thinspace(\textrm{mod }2\pi)\right\} \backslash\left\{ \left(0,0\right)\right\} \\
\Sigma_{0} = \Sigma_{F} & =\left\{ \left(\pi,\pi\right)\right\} \cup\left\{ \left(x_{1},x_{2}\right)\in\T^{2}~:~\,x_{1}=0~\textrm{or}~x_{2}=0\right\} \backslash\left\{ \left(0,0\right)\right\} \\
 & \Sigma^{g}=\Sreg\backslash\Sigma_{0}=\left\{ \left(x_{1},x_{2}\right)\in\T^{2}~:~\,x_{1}+x_{2}\equiv0\thinspace\thinspace(\textrm{mod }2\pi)\right\} \backslash\left\{ \left(0,0\right),\thinspace\left(\pi,\pi\right)\right\}.
\end{align*}
Note that $\Sigma_{0}$ is defined in (\ref{eq:zero_Sigma}) as a subset
of $\Sreg$ for which the corresponding eigenfunctions vanish at some
vertex. Also, $\Sigma_{F}$ is defined in (\ref{eq:Sigma_faces}) as a subset of $\Sigma{0}$, for which the corresponding eigenfunctions are supported on a loop. See also Figure \ref{fig:Sigma_decomposition} which shows this decomposition of the secular manifold.

\begin{figure}
  (a)\includegraphics[scale=0.6]{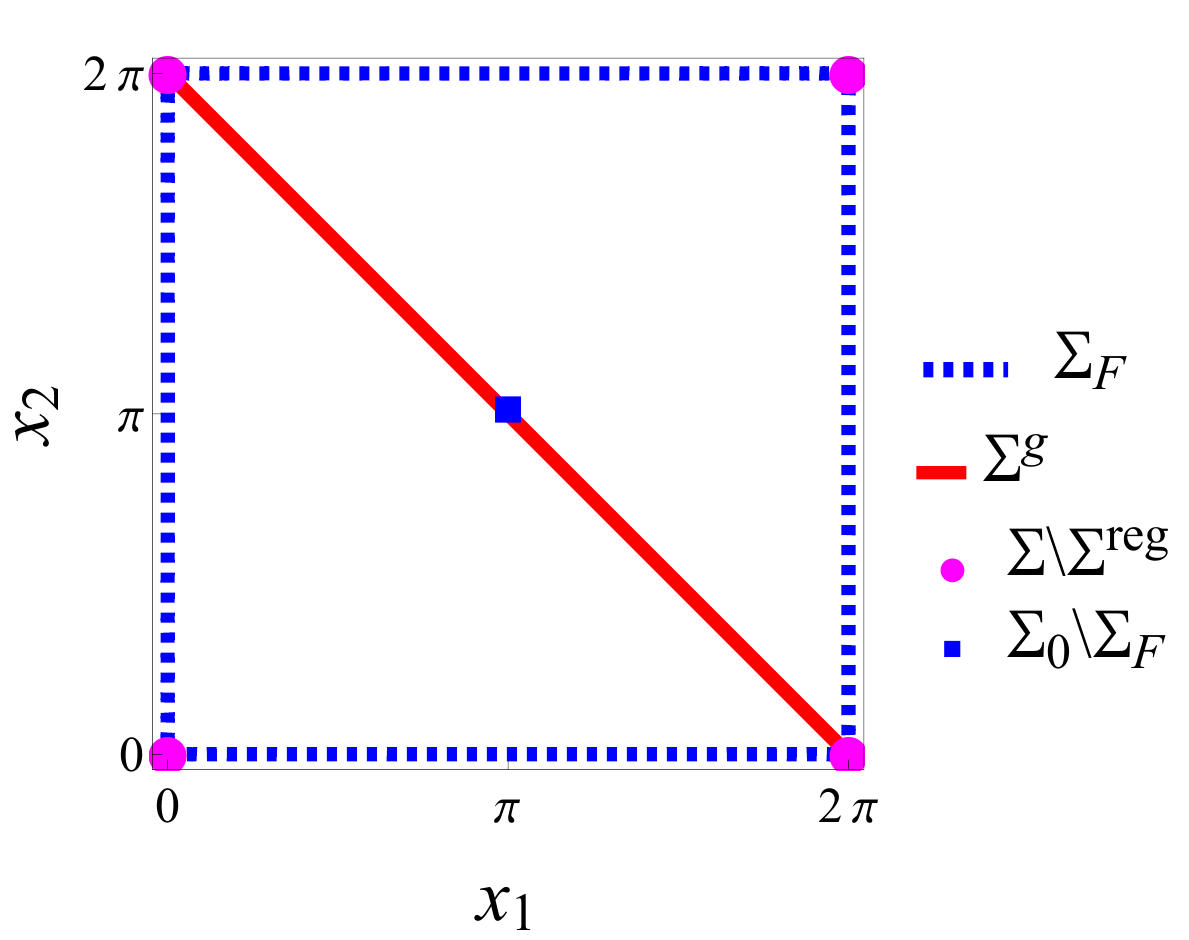}
  \qquad 
  (b)\includegraphics[scale=1.5]{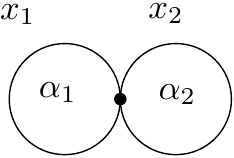}
  \caption{(a) The secular manifold of the ``figure of 8'' graph with
    different subsets ($\Sreg$, $\Sigma_{0}$, $\Sigma_{F}$ and
    $\Sgen$) highlighted by different colors. (b) The ``figure of 8'' graph.}
\label{fig:figure_8_torus}
\end{figure}

In addition, a straightforward calculation shows that for $\xv\in\Sigma^{g}$
(for which $\sin\left(\frac{x_{1}+x_{2}}{2}\right)=0$) we get
\[
\nabla_{x}F_{R}\cdot\lv=4\sin\left(\frac{x_{1}}{2}\right)\sin\left(\frac{x_{2}}{2}\right)\cos\left(\frac{x_{1}+x_{2}}{2}\right)\left(l_{1}+l_{2}\right)=-4\sin\left(\frac{x_{1}}{2}\right)\sin\left(\frac{x_{2}}{2}\right)\left(l_{1}+l_{2}\right)<0.
\]
Thus
\begin{align}
-\frac{H_{\alpha}F_{R}}{\nabla_{x}F_{R}\cdot\lv} & =-\frac{1}{4\sin\left(\frac{x_{1}}{2}\right)\sin\left(\frac{x_{2}}{2}\right)\left(l_{1}+l_{2}\right)}\begin{pmatrix}2\sin\left(x_{2}\right) & 0\\
0 & 2\sin\left(x_{1}\right)
\end{pmatrix},\\
\sigma\left(\xv\right) & =\M\left(-\frac{H_{\alpha}F_{R}}{\nabla_{x}F_{R}\cdot\lv}\right)=\M\begin{pmatrix}-\sin\left(x_{2}\right) & 0\\
0 & -\sin\left(x_{1}\right)
\end{pmatrix}
\end{align}
and as for $\xv\in\Sigma^{g}$ we have $x_{1}+x_{2}=2\pi$, this gives
\begin{equation}
\sigma\left(\xv\right)=\M\begin{pmatrix}\sin\left(x_{1}\right) & 0\\
0 & -\sin\left(x_{1}\right)
\end{pmatrix}\equiv1,\label{eq:surplus_of_figure_8}
\end{equation}
with the following local surplus functions
\begin{align*}
\forall\xv\in\Sigma^{g}\,\,,\sigma^{\left(i\right)}\left(x_{i}\right) & =\begin{cases}
0 & x_{i}\in\left(0,\pi\right)\\
1 & x_{i}\in\left(\pi,2\pi\right)
\end{cases}.
\end{align*}
It is obvious from (\ref{eq:surplus_of_figure_8}) that the surplus
distribution is not binomial. In addition, this example nicely demonstrates
that local surpluses may be anti correlated, $\Prc{\sigma^{(1)}=j}{\sigma^{(2)}=j}=0$.

\begin{remark} Another approach in this simple case, would be to
notice that for a given choice of rationally independent lengths $\lv$,
we have two kinds of eigenfunctions. The first kind are the eigenfunctions
supported on one of the loops and zero on the other, with corresponding
eigenvalues in $k\in\frac{2\pi}{l_{1}}\Z\cup\frac{2\pi}{l_{2}}\Z$.
The second kind of eigenfunctions can be obtained from an eigenfunction
of a circle of length $l_{1}+l_{2}$, under identification of two
points with the same function value which are at $l_{1}$ distance
apart. Such eigenfunctions, which are generic, would correspond to
eigenvalues $k\in\frac{2\pi}{l_{1}+l_{2}}\Z$ and their nodal count
will therefore be $\phi=\frac{k\left(l_{1}+l_{2}\right)}{\pi}$. It
remains to figure the position in the spectrum of such an eigenvalue.
One can check that the two sets $A=\frac{2\pi}{l_{1}}\Z\cup\frac{2\pi}{l_{2}}\Z$
and $B=\frac{2\pi}{l_{1}+l_{2}}\Z$ interlace and as the second eigenvalue
after $k_{1}=0$ will be $k_{2}=\frac{2\pi}{l_{1}+l_{2}}$, then we
get that the generic eigenvalues will be the even ones $k_{2n}=n\frac{2\pi}{l_{1}+l_{2}}$
with
\[
\phi_{2n}=\frac{k_{2n}\left(l_{1}+l_{2}\right)}{\pi}=2n~~\Rightarrow~~\sigma\equiv1
\]
\end{remark}

\subsection{A dumbbell graph}

\label{sec:dumbbell}

Consider the graph shown in Figure~\ref{fig:dumbbell_secular}(b)
(``dumbbell'' graph).

Its real secular function (\ref{eq:Real_sec_func}) can be calculated
to be
\begin{multline}
  F_{R}\left(x_{1},x_{2},x_{3};\alpha_{1},\alpha_{2}\right)
  = \frac{16}{9}\cos x_{2} \Big(\cos\alpha_{1}\sin
    x_{3}+\cos\alpha_{2}\sin
    x_{1}-\sin\left(x_{1}+x_{3}\right)\Big) \\
  -\frac{8}{9}\sin x_{2}\Big(4\left(\cos\alpha_{1}-\cos
      x_{1}\right)\left(\cos\alpha_{2}-\cos x_{3}\right)-\sin
    x_{1}\sin x_{3}\Big).
  \label{eq:FR_dumbbell}
\end{multline}

Observe that by taking the limit $x_{2}\rightarrow0$ we recover the
secular function of the figure of 8 graph, (\ref{eq:FR_figure8}),
up to a factor of $\frac{8}{9}$, which demonstrates the result of
Lemma \ref{lem:contracting_an_edge}.  

For $\av=0$ we get for the secular function
\begin{multline}
  F_{R}\left(x_{1},x_{2},x_{3};0,0\right) =\frac{16}{9}\sin\frac{x_{1}}{2}\sin\frac{x_{3}}{2}
  \left(-\sin x_{2}\left(3\cos\frac{x_{1}-x_{3}}{2}-5\cos\frac{x_{1}+x_{3}}{2}\right) \right. \\
  \left. +4\cos x_{2}\sin\frac{x_{1}+x_{3}}{2}\right)
  \label{eq:FR_dumbbell_zero_alpha}
\end{multline}
and correspondingly the Hessian is
\begin{multline}
  H_{\alpha}F_{R}\left(x_{1},x_{2},x_{3};0,0\right) \\
  = \frac{8}{9}\begin{pmatrix}4\sin x_{2}\left(1-\cos x_{3}\right)-2\cos x_{2}\sin x_{3} & 0\\
    0 & 4\sin x_{2}\left(1-\cos x_{1}\right)-2\cos x_{2}\sin x_{1}
  \end{pmatrix}.
\end{multline}
 We may use (\ref{eq:FR_dumbbell_zero_alpha}) in order to extract
$x_{2}\left(x_{1},x_{3}\right)$ for points $\xv\in\Sigma^{g}$ and
thus get the following expressions in terms of $x_{1},x_{3}$ solely,
\begin{equation*}
  -\frac{H_{\alpha}\left(F\right)}{\nabla F_{R}\cdot\lv}
  = \frac
  {
    \begin{pmatrix}
      4\sin\frac{x_{3}}{2}\cos\frac{x_{1}}{2}C(x_3) & 0\\
      0 & 4\sin\frac{x_{1}}{2}\cos\frac{x_{3}}{2}C(x_1)
    \end{pmatrix}
  }
  {2\sin\frac{x_{1}}{2}\sin\frac{x_{3}}{2}
    \left(2C(x_3)l_1 + C(x_3)C(x_1)l_2 + 2C(x_1)l_3\right)},\\
\end{equation*}
where $C(x) = 5-3\cos x$ is never zero.  Therefore,
\begin{equation*}
  \sigma\left(\xv\right) = 
  \M\left(-\frac{H_{\alpha}\left(F\right)}{\nabla
  F_{R}\cdot\lv}\right)
  = \M
  \begin{pmatrix}
    \cot\frac{x_{1}}{2} & 0\\
    0 & \cot\frac{x_{3}}{2}
  \end{pmatrix}.
\end{equation*}
By the above, the local surplus functions are given by
\begin{align*}
\sigma^{\left(1\right)}\left(x_{1}\right)=\begin{cases}
0 & x_{1}\in\left(0,\pi\right)\\
1 & x_{1}\in\left(\pi,2\pi\right)
\end{cases}, & \quad\quad\sigma^{\left(2\right)}\left(x_{3}\right)=\begin{cases}
0 & x_{3}\in\left(0,\pi\right)\\
1 & x_{3}\in\left(\pi,2\pi\right)
\end{cases}.
\end{align*}

\begin{figure}
  (a)\includegraphics[scale=0.75]{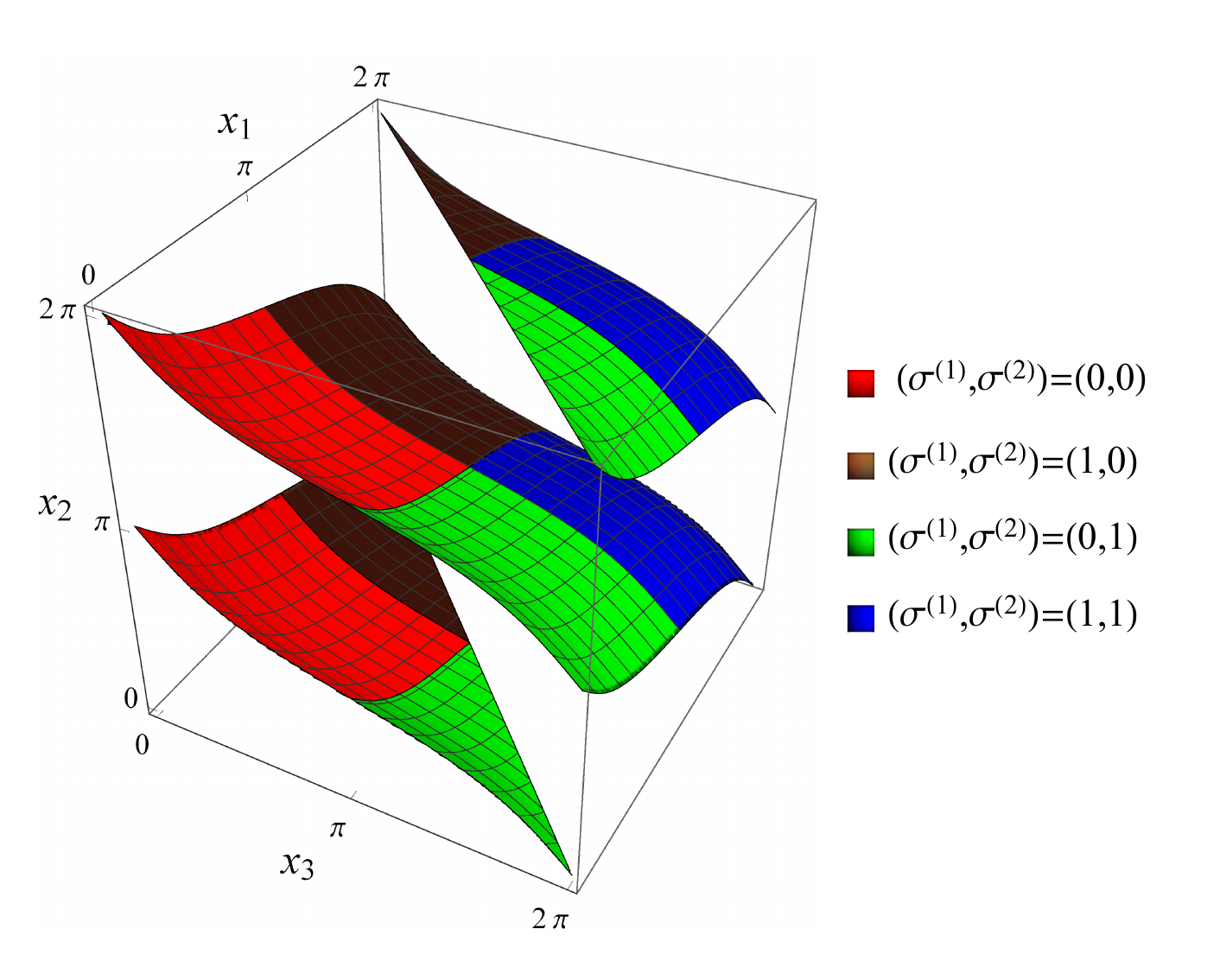} 
  (b)\includegraphics[scale=1.5]{dumbbell}
  \caption{The generic secular manifold $\Sigma^{g}$ for the
    ``dumbbell graph''.  It is depicted in four colors, according to
    all possible values for the local surpluses. (b) The ``dumbbell''
    graph with torus coordinates and magnetic fluxes marked on
    corresponding edges.}
\label{fig:dumbbell_secular}
\end{figure}

Figure~\ref{fig:dumbbell_secular}(a) shows the secular manifold of
the ``dumbbell'' graph with different local nodal surpluses indicated
by color. In Figure~\ref{fig:dumbbell_hist}, we give a normalized
histogram of the nodal surplus for the first $10^{6}$ eigenfunctions
calculated numerically for the rationally independent lengths $\lv=\left(\pi,e,1\right)$.
We compare it in the figure to the binomial distribution $\Bin\left(2,\frac{1}{2}\right)$
and find a perfect match according to the prediction of Theorem~\ref{thm:binomial_distr}.

\begin{figure}
\includegraphics[width=0.45\textwidth]{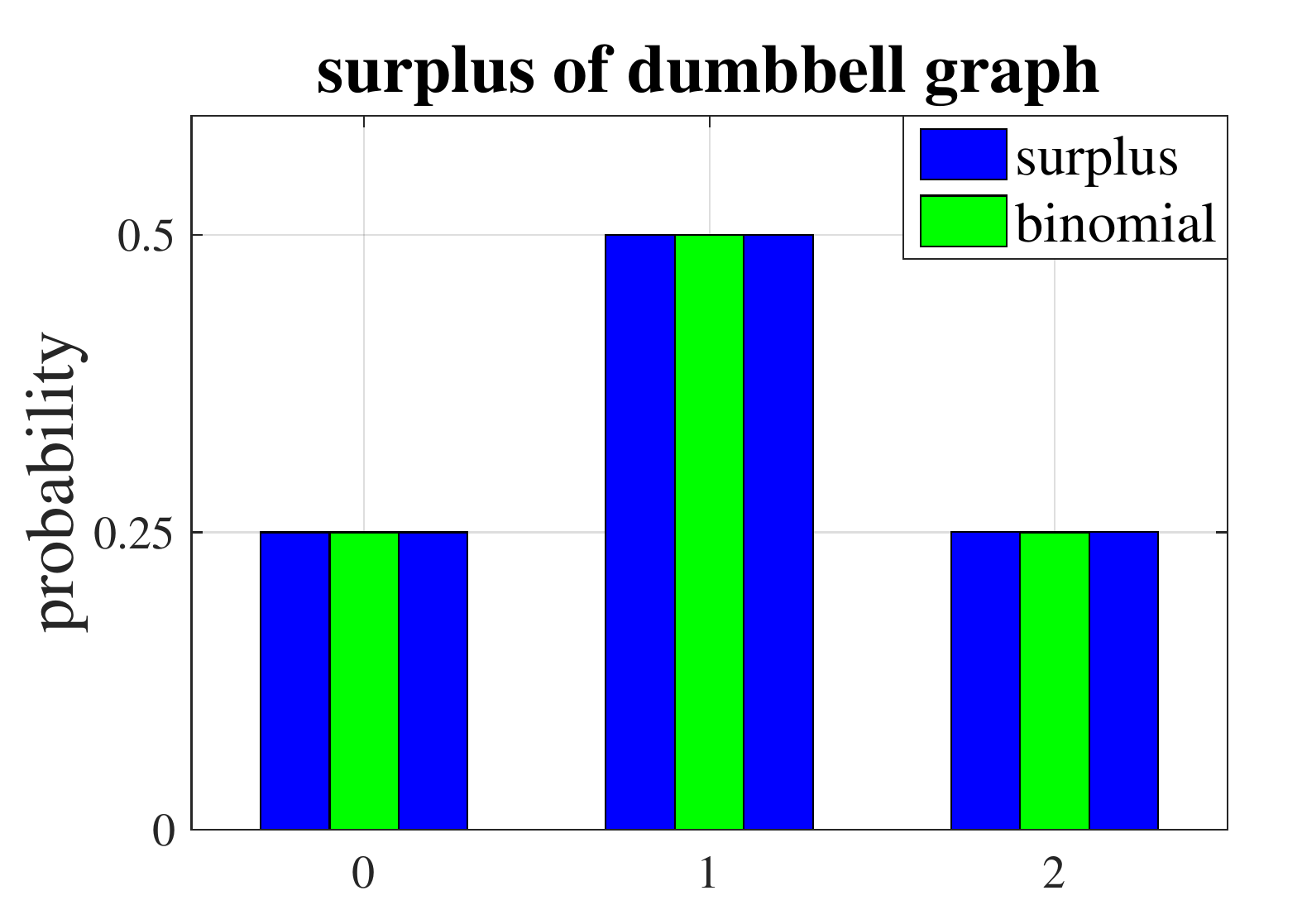} 
\caption{ The normalized histogram of the nodal surplus of the ``dumbbell''
graph calculated from the first $10^{6}$ eigenfunctions, compared
to the relevant binomial distribution.}
\label{fig:dumbbell_hist}
\end{figure}


\subsection{A $[1,2,2,1]$ pumpkin chain}

Consider the $[1,2,2,1]$ chain graph shown in Figure~\ref{fig:pumpkin1221_hist}(b).
In Figure~\ref{fig:pumpkin1221_hist}(a), we give a normalized histogram
of the nodal surplus for the first $10^{6}$ eigenfunctions calculated
numerically for the rationally independent lengths 
\begin{equation*}
  \lv=\left(\pi,e,1,\sqrt{2},\sqrt{3},\sqrt{5}\right).
\end{equation*}
We compare it in the figure to the binomial distribution $\Bin\left(2,\frac{1}{2}\right)$
and find that they match. This is in spite of the fact that this graph
does not satisfy the assumptions of Theorem~\ref{thm:binomial_distr},
as its cycles are not edge-separated. This numerical finding calls
for a further investigation.

\begin{figure}
  (a)\includegraphics[width=0.45\textwidth]{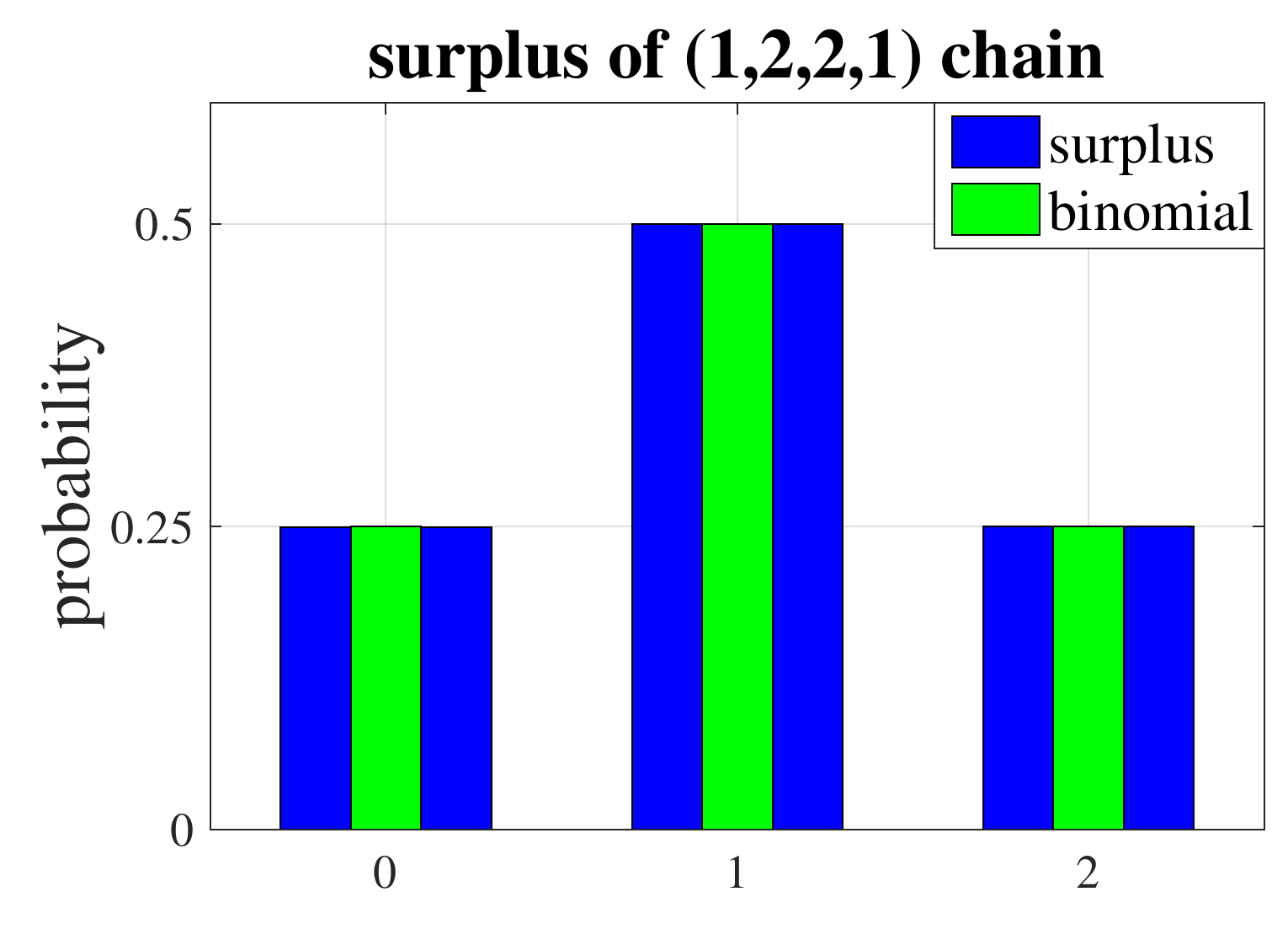}
  \qquad
  (b)\includegraphics[scale=1.5]{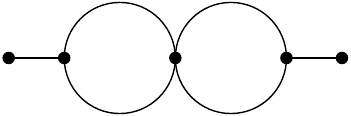} 
  \caption{(a) The normalized histogram of the nodal surplus of the
    $[1,2,2,1]$ pumpkin chain calculated from the first $10^{6}$
    eigenfunctions, compared to the relevant binomial
    distribution. (b) The $[1,2,2,1]$ pumpkin chain graph.}
\label{fig:pumpkin1221_hist}
\end{figure}


\subsection{A $[3,2,1]$ pumpkin chain}

\label{sec:321pumpkin}

Consider the $[3,2,1]$ chain graph shown in Figure~\ref{fig:pumpkin321_hist}(b).
In Figure~\ref{fig:pumpkin321_hist}(a), we show a normalized histogram
of the nodal surplus for the first $10^{6}$ eigenfunctions calculated
numerically for the rationally independent lengths $\lv=\left(\pi,e,1,\sqrt{2},\sqrt{3},\sqrt{5}\right)$.
We compare it in the figure to the binomial distribution $\Bin\left(3,\frac{1}{2}\right)$
as $\beta=3$ for this graph.

\begin{figure}
  (a) \includegraphics[width=0.45\textwidth]{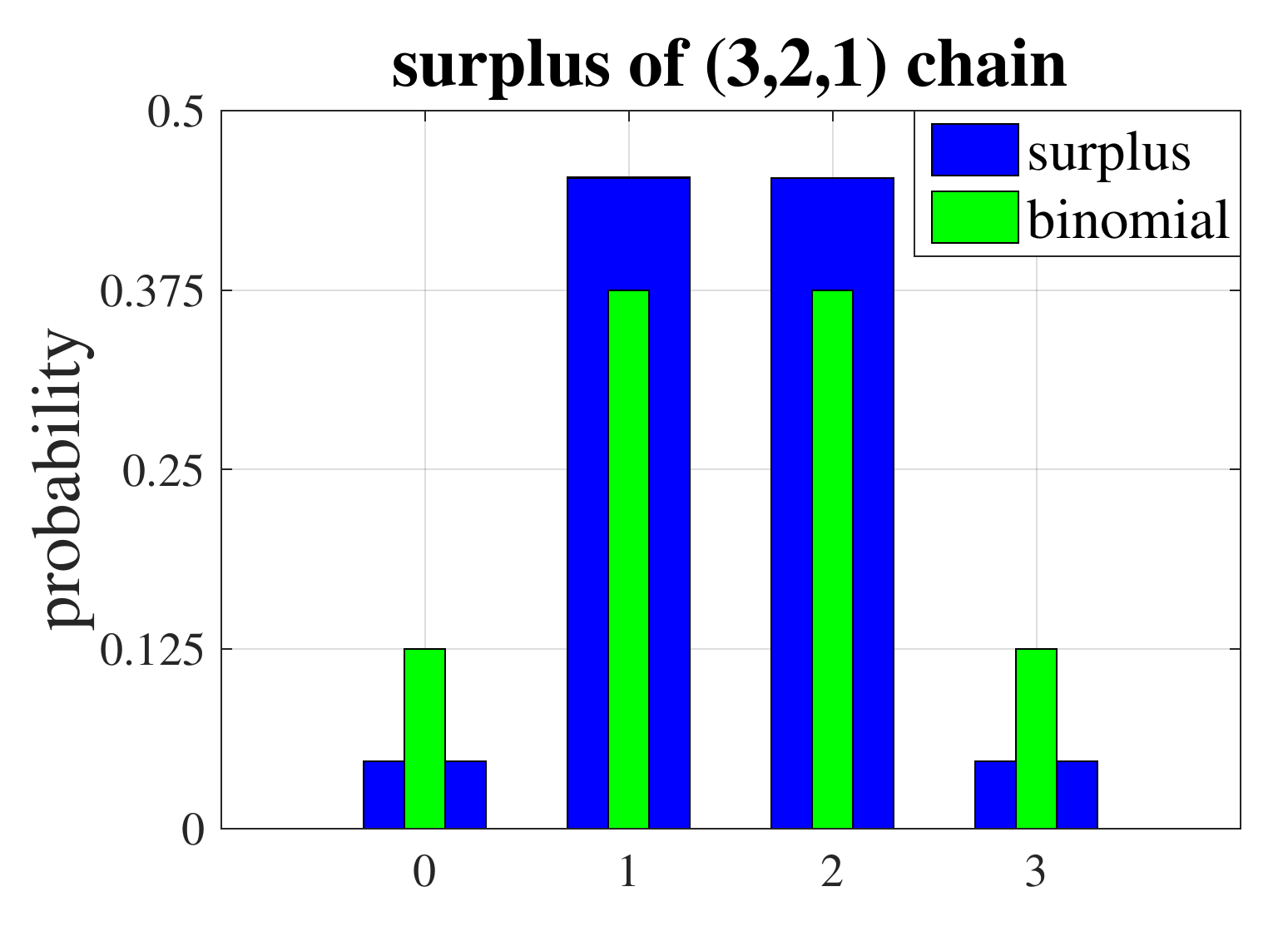}
  \qquad
  (b)\includegraphics[scale=1.5]{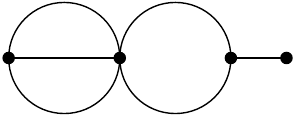}
  \caption{(a) The normalized histogram of the nodal surplus of the
    $[3,2,1]$ chain graph calculated from the first $10^{6}$
    eigenfunctions, compared to the relevant binomial
    distribution. (b) The $[3,2,1]$ chain graph.}
\label{fig:pumpkin321_hist}
\end{figure}

It is easy to notice that there is no match and the nodal surplus
probability of the $\left(3,2,1\right)$ chain graph is not binomial.
We further investigate this graph by examining the conditional probabilities
of its local surpluses. Note that this graph has two vertex-separated
blocks, of Betti numbers, $\beta^{(1)}=2,~\beta^{(2)}=1$. First,
we calculate numerically the conditional probability, $\Prc{\siloc{2}=s}{\siloc{1}}$
for different values of $\siloc{1}$ and see that it is not symmetric
(shown in Figure\,\ref{fig:/pumpkin321_conditional_prob_hist}(a)).
Then, we do the same for the other conditional probability, $\Prc{\siloc{1}=s}{\siloc{2}}$ in Figure\,\ref{fig:/pumpkin321_conditional_prob_hist}(b))
and once again find no symmetry. This demonstrates that the property
of independently symmetric local surpluses (see Theorem \ref{thm:indep_symmetric})
does not hold for this graph.

\begin{figure}
(a)\includegraphics[width=0.45\textwidth]{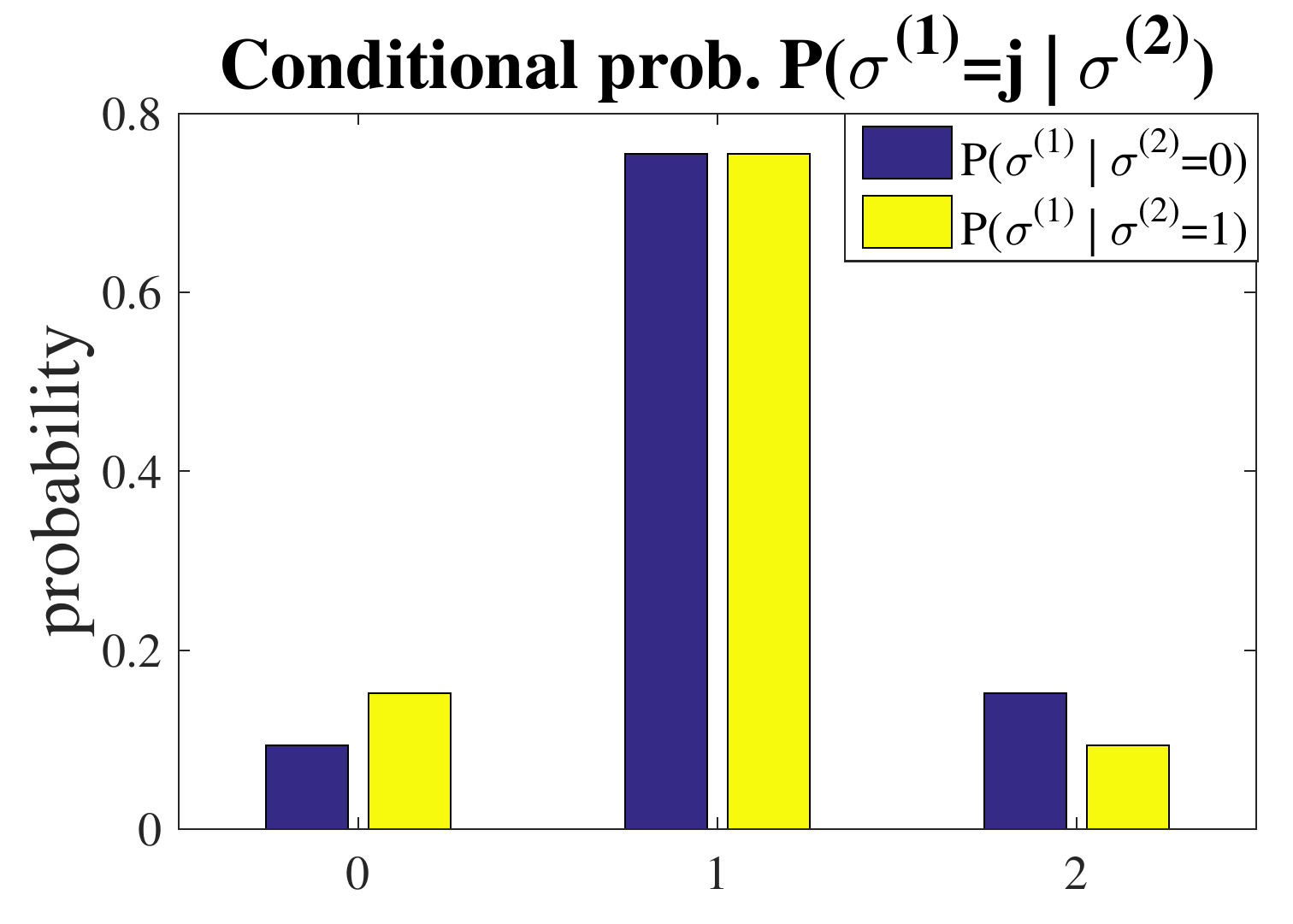}\hfill{}
(b)\includegraphics[width=0.45\textwidth]{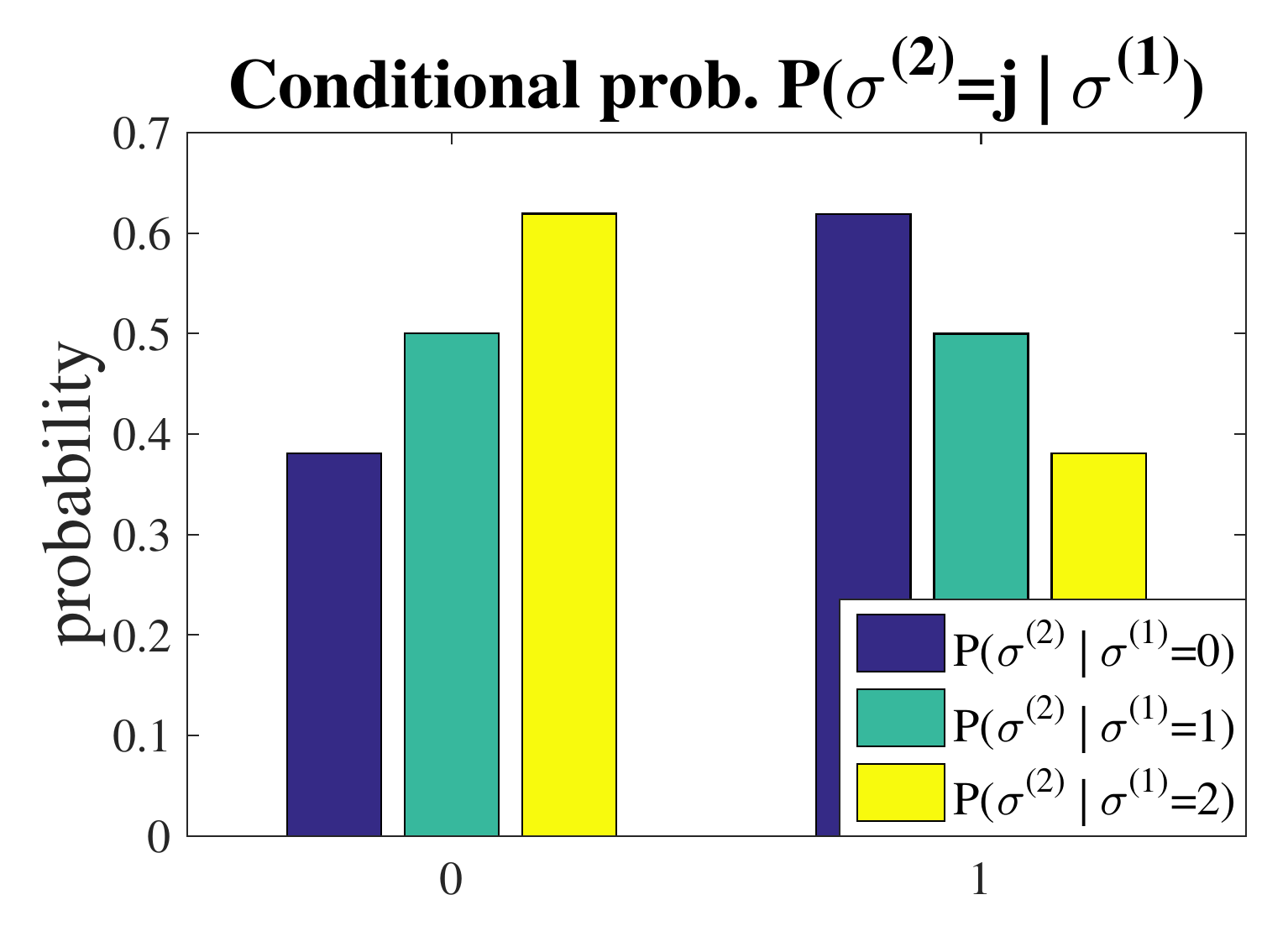}
\caption{The normalized histograms of the conditional probability of
  local surpluses for the graph in
  Figure~\ref{fig:appendix_examples}(d) calculated from the first
  $10^{6}$ eigenfunctions.  (a) $\Prc{\siloc{1}=j}{\siloc{2}}$, (b)
  $\Prc{\siloc{2}=j}{\siloc{1}}$ }
\label{fig:/pumpkin321_conditional_prob_hist}
\end{figure}


\section*{Acknowledgment}

The collaboration that made this project possible was supported, in
part, by the Binational Science Foundation Grant (Grant
No.~2016281).  GB was partially supported by NSF grant DMS-1410657. RB
and LA were supported by ISF (Grant No.~494/14). RB was supported by
Marie Curie Actions (Grant No.~PCIG13-GA-2013-618468).


\def\cprime{$'$}

\end{document}